\documentclass[a4paper,USenglish,cleveref, autoref, thm-restate]{lipics-v2021}

\usepackage{xspace}
\usepackage{dsfont}
\usepackage{thmtools}
\usepackage{mismath}
\usepackage{bbm}
\usepackage{mathtools}
\usepackage{ifthen}
\usepackage{comment}
\usepackage{parskip}
\graphicspath{{img/}}
\usepackage{graphicx}
\usepackage{amsmath,amsfonts,amssymb}
\usepackage{mathrsfs}
\usepackage{thm-restate}
\usepackage{wrapfig}
\usepackage{cases}
\usepackage{xspace}
\usepackage{csquotes}
\usepackage{multirow}
\usepackage{multicol}
\usepackage{mathtools}

\usepackage{footmisc}

\hideLIPIcs  %

\newcommand{\sv}[1]{}%
 \newcommand{\lv}[1]{#1}%
\usepackage{etoolbox}
\newcommand{\appendixText}{}

 \newcommand{\toappendix}[1]{#1}%

\newcommand{\Nn}{\mathbb{N}}
\newcommand{\Graph}{\ensuremath{G}}

\newcommand{\VS}{\ensuremath{X}}
\newcommand{\VSx}{\ensuremath{X}}
\newcommand{\MS}{\ensuremath{D}}
\newcommand{\Vertices}{\ensuremath{V}}

\newcommand{\GraphInd}[1]{\ensuremath{\Graph[{#1}]}}
\newcommand{\nb}[1]{\ensuremath{N(#1)}}
\newcommand{\nbd}[2]{\ensuremath{N_{#1}(#2)}}

\newcommand{\cvd}{\ensuremath{\mathop{\mathrm{cvd}}}}

\newcommand{\FPT}{\ensuremath{\mathsf{FPT}}\xspace}
\newcommand{\XP}{\ensuremath{\mathsf{XP}}\xspace}
\newcommand{\W}[1]{\ensuremath{\mathsf{W[#1]}}}

\newcommand{\NP}{{\ensuremath{\mathsf{NP}}}\xspace}

\newcommand{\cliFC}{\ensuremath{\mathsf{CliqueFC}}\xspace}
\newcommand{\MSs}{D_G}

\newcommand{\ctypes}{\ensuremath{\mathsf{CT}}}
\newcommand{\ntypes}{\ensuremath{\mathsf{NT}}}
\newcommand{\spatterns}{\ensuremath{\mathsf{SP}}}
\newcommand{\shapes}{\ensuremath{\mathsf{Shapes}}}
\newcommand{\shapefn}{\ensuremath{\mathsf{Shape}}}

\newcommand{\fc}[1]{\ensuremath{\mathsf{fc}(#1)}}
\newcommand{\fcnobrackets}[1]{\ensuremath{\mathsf{fc}}}

\newcommand{\Cc}{\mathcal{C}}

\newcommand{\Xx}{\mathcal{X}}
\newcommand{\setupto}[1]{\ensuremath{[#1]}}

\newcommand{\cT}{\ensuremath{\mathsf{cT}}\xspace}
\newcommand{\ctype}{\ensuremath{\mathsf{cT}}\xspace}
\newcommand{\nT}{\ensuremath{\mathsf{nT}}\xspace}
\newcommand{\sP}{\ensuremath{\mathsf{sP}}\xspace}

\newcommand{\mtx}{\ensuremath{\mathsf{M}}\xspace}
\newcommand{\shape}{\ensuremath{\mathsf{shp}}\xspace}

\newcommand{\nodd}{\ensuremath{\mathbb{N}_\text{odd}}\xspace}

\newcommand{\FO}{{$\mathsf{FO}$}\xspace}

\newcommand{\MSOt}{{\ensuremath{\mathsf{MSO}_2}}\xspace}
\newcommand{\MSOo}{\ensuremath{\mathsf{MSO}_1}\xspace}

\newcommand{\dbinpacking}{\textsc{Unary $\ell$-Bin Packing}\xspace}
\newcommand{\dtuple}{\textsc{Unary $d$-Tuple}\xspace}
\newcommand{\FairVD}{\textsc{Fair Vertex $\mathsf{FO}$ Deletion}\xspace}
\newcommand{\FairVE}{\textsc{FairVE}\xspace}

\newcommand{\df}{\coloneqq}
\newcommand{\NOT}[1]{\overline{#1}}
\def\phi{\varphi}
\def\tilde{\widetilde}

\usepackage{xcolor}

\newcommand{\TM}[1]{\textcolor{orange}{TM: #1}}

\newcommand{\prob}[3]{
\begin{center}
\renewcommand{\arraystretch}{0.85}%
\begin{tabularx}{\textwidth}{|lX|}
	\hline
	\multicolumn{2}{|l|}{\hspace{0pt}\rule{0pt}{13pt} #1}\\
	{\bf Input:\enspace}&{#2}\\
	{\bf Question:\enspace}&{#3 \rule[-6pt]{0pt}{6pt}} \\
	\hline
\end{tabularx}
\end{center}
}

\title{Fair Vertex Problems Parameterized by Cluster Vertex Deletion}

\author{Tom\'a\v{s} Masa\v{r}\'ik}{Institute of Informatics, Faculty of Mathematics, Informatics and Mechanics, University of Warsaw}{masarik@mimuw.edu.pl}{0000-0001-8524-4036}{Supported by the Polish National Science Centre SONATA-17 grant number
2021/43/D/ST6/03312.}
\author{Jędrzej Olkowski}{Institute of Informatics, Faculty of Mathematics, Informatics and Mechanics, University of Warsaw}{jo417777@students.mimuw.edu.pl}{TODO ORCID}{Supported by the Polish National Science Centre SONATA-17 grant number 2021/43/D/ST6/03312.}
\author{Anna Zych-Pawlewicz}{Institute of Informatics, Faculty of Mathematics, Informatics and Mechanics, University of Warsaw}{anka@mimuw.edu.pl}
{https://orcid.org/0000-0002-5361-8969}{\flag[0.17\textwidth]{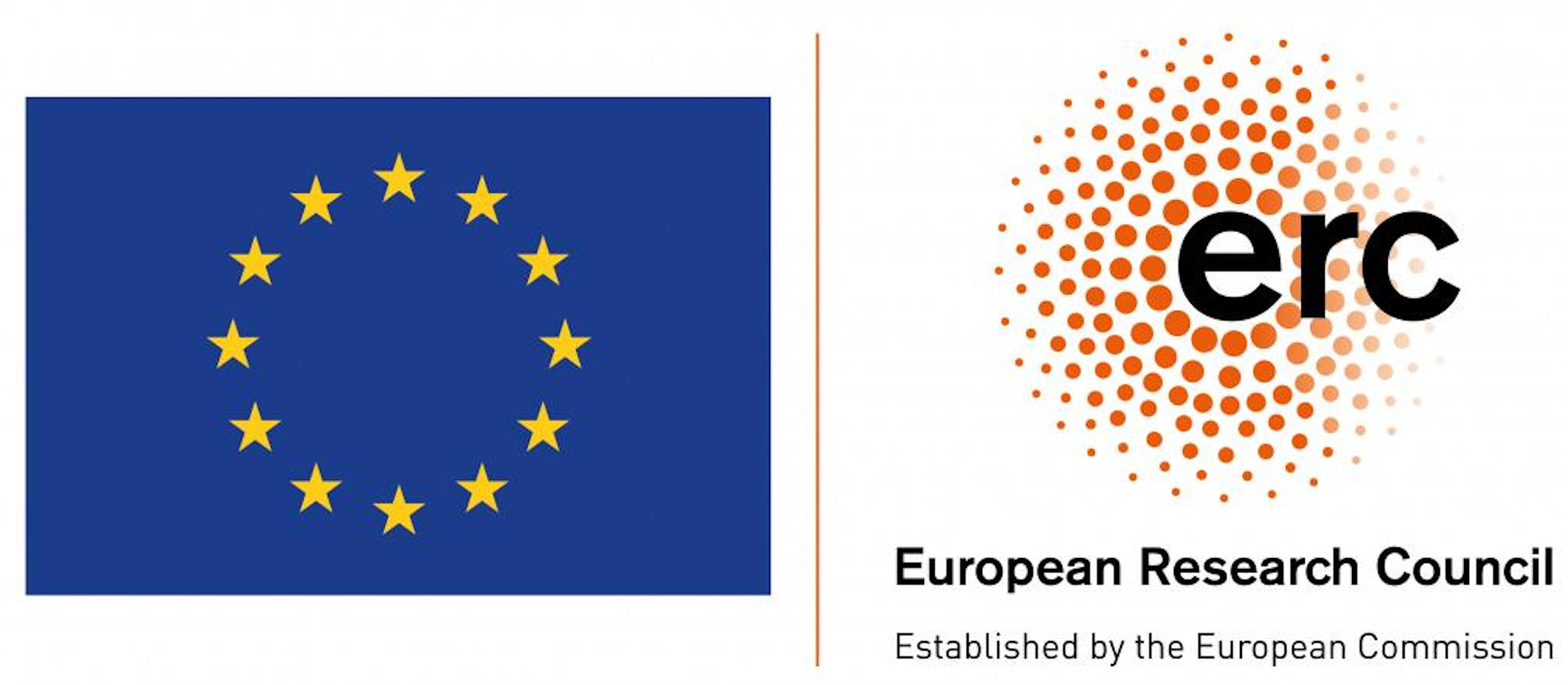}This work is a part of project BOBR that has received funding from the European Research Council (ERC) under the European Union’s Horizon 2020 research and innovation programme (grant agreement No. 948057). }

\authorrunning{T. Masa\v{r}\'ik, J. Olkowski, A. Zych-Pawlewicz}

\Copyright{T. Masa\v{r}\'ik, J. Olkowski, A. Zych-Pawlewicz}



\category{} %

\sv{\relatedversiondetails[cite=ourArxiv]{Full Version}{https://arxiv.org/abs/2502.01400}} %

\nolinenumbers %

\EventEditors{John Q. Open and Joan R. Access}
\EventNoEds{2}
\EventLongTitle{42nd Conference on Very Important Topics (CVIT 2016)}
\EventShortTitle{CVIT 2016}
\EventAcronym{CVIT}
\EventYear{2016}
\EventDate{December 24--27, 2016}
\EventLocation{Little Whinging, United Kingdom}
\EventLogo{}
\SeriesVolume{42}
\ArticleNo{23}

\begin{document}

\maketitle

\begin{abstract}
In this paper we study fair variants of \MSOo definable problems parameterized by cluster vertex deletion number, i.e., the smallest number of vertices required to be removed from the graph such that what remains is a collection of cliques.
While typical graph problems seek the smallest set of vertices satisfying some property,
their fair variants seek such a set that does not contain too many vertices in any neighborhood of any vertex.
Formally, the task is to find a set $X\subseteq V(G)$ satisfying some \MSOo definable property, whose fair cost is at most $k$, i.e., such that for all $v\in V(G)$ it holds that $|X\cap N(v)|\le k$,
where $N(v)$ stands for an open neighborhood of a vertex $v$.
Recently, Knop, Masařík, and Toufar [MFCS 2019] showed that all fair \MSOo definable problems can be solved in \FPT time parameterized by the twin cover of a graph.
They asked whether such a statement can be achieved for a more general parameterization by cluster vertex deletion number.

In this paper, we prove that in full generality this is not possible
by demonstrating \W{1}-hardness.
On the other hand, we give a sufficient condition under which a
fair \MSOo definable problem admits an \FPT algorithm parameterized by
the cluster vertex deletion number.
Our algorithm is general enough to capture the fair variant
of many natural graph problems such as the Fair Feedback Vertex Set problem, the Fair Vertex Cover problem, the Fair Dominating Set problem, the Fair Odd Cycle Transversal problem, as well as connected variants thereof.
Moreover, we solve %
the Fair $[\sigma,\rho]$-Domination problem for $\sigma$ finite, or when both $\sigma$ and $\rho$ are cofinite.
That is, given finite or cofinite $\rho,\sigma\subseteq \mathbb{N}$, 
the task is to find set of vertices $X\subseteq V(G)$ of fair cost at most $k$ such that for all $v\in X$, $|N(v)\cap X|\in\sigma$ and for all $v\in V(G)\setminus X$, $|N(v)\cap X|\in\rho$.
\end{abstract}

\newpage

\section{Introduction}\label{sec:intro}

We investigate the parameterized complexity of a rich class of graph problems under structural parameterization.
In particular, we study fair problems that are becoming more popular in recent years, although their history can be traced back to the 1980s~\cite{LinSahni}.
The paradigm of the \emph{fair problems} is that we seek a solution, usually represented as a subset of some universe, which is well-distributed within the given underlying structure but does not necessarily have the smallest possible size.
The standard formulation of the idea above, which is tailored to graph problems, uses the fair cost function as a measure of well-distributedness.
Given a graph $G$, we define the \emph{fair cost} (\fcnobrackets{}) of a set $\VS \subseteq \Vertices(\Graph)$ as $\fc{\VS}\df \max_{v \in \Vertices(\Graph)} |\nb{v} \cap \VS|$, where $\nb{v}$ is the open neighborhood of a vertex $v$.
Then the \emph{fair vertex problems} on graphs ask for a set $\VS\subseteq V(G)$ of the smallest fair cost among all vertex subsets satisfying some property.
As we wish to solve a large collection of problems, we need to formally characterize what kind of properties we work with.
For that, we use the concept of graph logic which is a standard tool that allows us to specify a collection of problems
definable using some particular graph logic.
For example, in the \MSOo logic\footnote{The logic \FO allows quantifications only over elements (vertices or edges).
  A more powerful logic \MSOo extends \FO by permitting quantification over sets of vertices (but not sets of edges).
For more details consult \cite{CE-book}.
}, for which we design our main \FPT algorithm, we can define many standard graph properties such as $k$-colorability, acyclicity, and even the connectivity of a set. 
Following~\cite{KMT19}, we provide two definitions of fair problems, the
second one being more general.

\prob{\sc Fair Vertex $\mathsf{L}$ Deletion ($\mathsf{L}$-FairVD) problems}
{An undirected graph $\Graph$, $k\in\mathbb{N}$, and a sentence $\varphi$ in logic \textsf{L}.}
{Is there a set $\VS \subseteq \Vertices(G)$ of fair cost at most $k$ such that $\GraphInd{\Vertices(G) \setminus \VS} \models \varphi$?}

In the above definition, $\GraphInd{Y}$ denotes a subgraph of
$\Graph$ induced by the set of vertices $Y \subseteq \Vertices(G)$.
This definition does not capture the Fair Dominating Set
problem, i.e., a problem asking for a set $X\subseteq V(G)$
such that each vertex of $G$ either belongs to $X$
or has a neighbor in it.
Hence, we present a more powerful variant where we allow
the formula to use free variables.
Note that even if $\mathsf{L}$ does not inherently have access to additional set variables, the formula can still use the free set variables $X_1,\ldots, X_\ell$.

\prob{\sc $\ell$-Fair Vertex $\mathsf{L}$ Evaluation problems}
{An undirected graph $\Graph$, $k\in\mathbb{N}$, and a formula $\varphi(X_1,\ldots,X_\ell)$ with $\ell$ free variables in logic \textsf{L}.}
{Are there sets $W_1,\ldots,W_\ell$ of vertices each of fair cost at most $k$ such that $\Graph \models \varphi(W_1,\ldots,W_\ell)$?}{}

In the special case of $\ell=1$, we refer
to $\ell$-Fair Vertex $\mathsf{L}$ Evaluation problems
as {\sc $\mathsf{L}$-FairVE} problems.
So far, we defined the fair vertex problems when provided with logic $\mathsf{L}$ and graph $G$, where the aim is to find a set $X\subseteq V(G)$ of the smallest fair cost that satisfies some formula of $\mathsf{L}$.
To speak about particular problems, we often use a template
where we say \emph{fair $\mathcal{P}$ problem} if the
problem $\mathcal{P}$ can be described in some logic $\mathsf{L}$ and, therefore, its fair variant belongs to the {\sc $\mathsf{L}$-FairVE} or even to the {\sc $\mathsf{L}$-FairVD} problems.
An important example
is the Fair Vertex Cover problem, whose deletion variant can be
described in \FO logic by a simple formula stating that
the graph is edgeless, i.e.,
for all vertices $x,y$ we have $xy\not\in E(G)$.
Therefore, the Fair Vertex Cover problem belongs to the {\sc \FO-FairVD} problems.

We can easily modify the fair vertex problem definitions to fair edge problems where we seek $F\subseteq E(G)$ subject to the \emph{edge fair cost} function defined as $\fc{F}\df \max_{v \in \Vertices(\Graph)} \text{deg}_F(v)$, where $\text{deg}_F(v)$ measures the number of edges incident to a vertex $v$ that are in $F$.
The edge variant was the originally studied variant in early works~\cite{Kolman09onfair,KLS10Kam,LinSahni} until the vertex variant was formally defined in 2017~\cite{MT20}. %
\sv{For more details on the edge variant see~\cite{KMT19}.}
\lv{Edge variant is important for the context, but in the current paper we only explore the fair vertex problems and, therefore, we refrain from formal definitions.
For more details on the topic see~\cite{KMT19}. %
}

In this paper, we study the fair vertex problems under structural parameterization.
Very recently, in 2023, Gima and Otachi~\cite{GimaOtachi} showed that {\sc $\ell$-Fair Vertex \MSOt Evaluation} is \FPT parameterized by the vertex integrity and the size of the formula. 
This is complemented by \W{1}-hardness of the Fair Vertex Cover problem even when parameterized by the combined treedepth and feedback vertex set number~\cite{KMT19}.
Since there is very little gap between those parameters, the situation on the sparse side of the spectrum seems to be relatively understood.
On the other hand, we know much less about dense graph parameters.
Knop, Koutecký, Masařík, and Toufar~\cite{KKMT} showed that {\sc $\ell$-Fair Vertex \MSOo Evaluation} is \FPT parameterized by the neighborhood diversity and the size of the formula. 
Note that \MSOo cannot be extended to \MSOt on dense graph parameters because \MSOt model-checking on cliques is not even in \XP unless $\mathsf{E} = \mathsf{NE}$ \cite{CMR,Lampis14}.
Parameterization by the twin cover is explored in \cite{KMT19} and the {\sc \MSOo-FairVE} is \FPT parameterized by the twin cover and the size of the formula while {\sc $\ell$-Fair Vertex \MSOo Evaluation} is already \W{1}-hard under the same parameterization.
The question raised in~\cite{KMT19} was whether the {\sc \MSOo-FairVE} remains \FPT even when parameterized by the cluster vertex
deletion number, which is a stronger parameter. %
The \emph{cluster vertex deletion} number is the size of the smallest set $D\subseteq V(G)$ such that $G-D$ (i.e., $G[V(G)\setminus D]$) is a collection of disjoint cliques.
We denote a set $D\subseteq V(G)$ such that $G-D$ is a collection of disjoint cliques a \emph{modulator} of the graph $G$.
The cluster vertex deletion number is a well-known optimization problem; see e.g., \cite{IKP23} for a very recent overview of results.
In particular, it has already been successfully used as a structural graph parameter to provide \FPT algorithms for various graph problems in the past two years~\cite{BKR23,GKKMV23,KS23}.
In this paper, we resolve the general question negatively, while providing a metatheorem capturing
an extensive collection of problems that do admit an \FPT algorithm under our parameterization.
We refer to \cref{fig:classes} for a visual overview of the results and graph class relations.

\begin{figure}[t]
\centering
\begin{subfigure}{0.32\textwidth}
    \includegraphics[width=\textwidth]{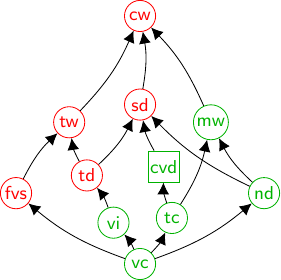}
    \caption{Parameterized complexity of the Fair Vertex Cover problem.}
    \label{fig:classes:VC}
\end{subfigure}
\hfill
\begin{subfigure}{0.32\textwidth}
    \includegraphics[width=\textwidth]{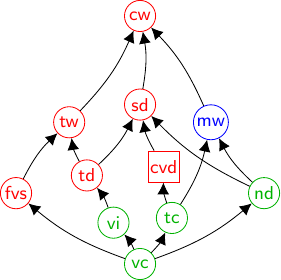}
  \caption{Parameterized complexity of the {\sc \MSOo-FairVE} problem.}
    \label{fig:classes:FairVE}
\end{subfigure}
\hfill
\begin{subfigure}{0.32\textwidth}
    \includegraphics[width=\textwidth]{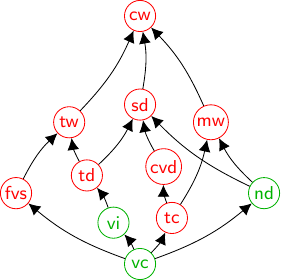}
    \caption{Parameterized complexity of the {\sc $\ell$-Fair Vertex \MSOo Evaluation} problem.}
    \label{fig:classes:FairlVE}
\end{subfigure}
        
\caption{Parameterized complexity of the fair vertex problems. The green color means there is an \FPT algorithm parameterized by the structural parameter (and the size of the formula if it applies). The red color highlights parameters for which there is a \W{1}-hardness result, and the blue color marks parameters that are yet unknown.
Results bounded in rectangles are new in this paper.
An arrow indicates that the parameter at the tip is bounded by a function of the parameter at the tail.
Abbreviations: $\mathrm{vc}$ = vertex cover number, $\mathrm{fvs}$ = feedback vertex set number,
$\mathrm{td}$ = treedepth, $\mathrm{tw}$ = treewidth, $\mathrm{cw}$ = clique-width, $\mathrm{vi}$ = vertex integrity,
$\mathrm{tc}$ = twin cover number, $\mathrm{nd}$ = neighborhood diversity, $\mathrm{mw}$ = modular-width,
$\mathrm{cvd}$ = cluster vertex deletion number, and $\mathrm{sd}$ = shrub-depth.
}
\label{fig:classes}
\end{figure}

\lv{\paragraph*{Related Results.}}%
There was also some interest in exploring fair
problems under parameterizations other than structural.
Parameterization by the solution size was studied in~\cite{JRS,KMMS21}.
Papers~\cite{IKKPS25,JRS,FairHighwayDimension} studied an
extension of the fair concept to sets, specifically %
the fair hitting set problem under various parameterizations.
In particular, \cite{FairHighwayDimension} shows that the Fair Vertex
Cover problem is \NP-complete for any $k\ge 4$ even on planar
graphs, where $k$ is the fair cost of the solution. %

The concept of Fair Problems has been generalized using the notion of \emph{local cardinality constraints}, as introduced by Szeider~\cite{Szeider:11}. 
Given a graph $G$ with cardinality constraints $\alpha(v)\subseteq [n]$ for all $v\in V(G)$,
a solution is a subset $W\subseteq V(G)$ such that $|W\cap N(v)|\in \alpha(v)$ for all $v\in V(G)$. A notable special case, introduced in~\cite{KKMT}, is called \emph{local linear cardinality constraints}, where all $\alpha(v)$ form an interval.
Local linear problems were further studied in~\cite{Knop2021}.
Observe that the fair vertex problems with fair cost $f$ represent the simplest local cardinality constraints, where $\alpha(v)=[0,f]$ for all $v\in V(G)$.

\subsection{Overview of Our Results}\label{sec:overview}

Here, we give an overview of results.
First, we state our hardness theorem that answers a general case of a question posted in~\cite{KMT19}. %
Notice that the hardness is for the more restricted deletion variant where only a sentence in \FO logic is used after the deletion of a set.
\sv{The proof of the following theorem is postponed to \cref{sec:hardness}.\footnote{We mark the statements that are proved in the appendix by $\clubsuit$. %
}}
\lv{\begin{restatable}[Hardness]{theorem}{WhThm}\label{thm:hardness} }
  \sv{\begin{restatable}[Hardness $\clubsuit$]{theorem}{WhThm}\label{thm:hardness}}
    The Fair Vertex \FO Deletion is \W{1}-hard parameterized by the size of the formula and the cluster vertex deletion number, for any cluster vertex deletion number greater than 3.
\end{restatable}

To state the positive result formally, we need to specify
some conditions under which it holds, as by Theorem~\ref{thm:hardness} it cannot hold in full generality.
For that we need to
introduce more model-checking background and tools which we do
in \Cref{sec:mc}.
Here, for the sake of simplicity, we give an overview,
intuition, and a less formal statement.
We also explain the main idea behind the methods of our proof.
We conclude with the list of fair problems which, by the power
of our positive result,
are solvable in \FPT time parameterized by the cluster vertex deletion number.

A big-picture strategy of our proof can be summarized as follows. 
First, we assign a \emph{shape} to every pair $(G,X)$
where $G$ is a graph and $X \subseteq V(G)$ represents
a potential solution to our problem
(See Section~\ref{sec:shape} for more details).
The shapes can be viewed as equivalence classes from the perspective of the formula.
That means we can check the truthfulness of the formula on the shape.
A desired outcome of this part of the proof is that the number of
possible shapes is bounded by the parameters: the formula size and the size of a modulator.
Therefore, all the shapes can be enumerated in \FPT time.
A drawback of this strategy is that the fair cost of the potential solution $X$ cannot be easily mapped to the shapes and has to be computed.
In \Cref{sec:FPT} we show how to compute the actual solution of some
particular shape with
a small fair cost.
More precisely, given a shape we are able to formulate an
integer linear program (ILP) to recover the solution of minimal fair cost, which belongs to the equivalence class of this shape.
Moreover, such an ILP has the number of variables bounded by our parameters,
therefore we utilize \sv{an}\lv{the following} algorithm originally proposed by Lenstra~\cite{Lenstra83}.
\begin{theorem}[{\cite{Lenstra83} with an improved running time by Reis and Rothvoss~\cite{RR23}}]\label{thm:ILPinBoundedDimension}
  There is an algorithm that, given an ILP (Integer Linear Program) with $p$ variables and $m$ constraints, finds an optimal solution in time ${(\log p)}^{\mathcal{O}(p)} \cdot (m\log L)^{\mathcal{O}(1)}$, where $L$ is the maximum absolute value of the coefficients.
\end{theorem}

Hence, the approach above leads to a solution in \FPT time.
This general approach was used in both~\cite{KMT19, MT20} to
obtain similar general results for fair problems
under different parameterization such as twin cover.

Although the general approach to the problem is similar,
parameterization by cluster vertex deletion poses a number of
new challenges. First of all, the vertices in $G - D$ (where $D$
is the modulator) may have different neighborhoods in $D$,
whereas under the parameterization by twin cover all vertices in the same clique of
$G-D$ are twins. Second of all, as mentioned before,
parameterization via cluster vertex deletion is
\W{1}-hard. Hence, we need to apply careful analysis to find
an assumption, as
general as possible, that makes the problem \FPT.

For our \FPT algorithm to work we need a somewhat consistent behavior of a fair cost solution on cliques in $G - D$. An example of such behavior is the necessity to include all but one vertex in the solution from each clique, as in the Fair Vertex Cover problem. Another example is the existence of a low fair cost solution that contains only a bounded (in terms of parameters) number of vertices from each clique.
We now (somewhat informally) explain the general assumption that we
propose. When working with a modulator $D$ it is convenient to
partition all cliques in $G-D$ into \emph{clique types}. We divide cliques into types depending on their neighborhoods in $D$. To be more precise, a clique type is a vector indexed with subsets $S \subseteq D$, storing a number for each $S \subseteq D$. This number describes how many vertices of the clique have neighborhood $S$ in $D$. Since we want to enumerate all possible clique types,
storing the exact number of such vertices is prohibitively
expensive. Thus, the clique type stores the precise number only
up to a certain threshold $\alpha$. Above this threshold we only
know from the clique type that there is more than $\alpha$ vertices
in the clique with neighborhood $S$.
In this way, we can partition all cliques in $G - D$ into a bounded number of clique types.

We consider some solution $X \subseteq V(G)$ to our
problem that we wish to find. Since we can enumerate all
possible candidates for $X \cap D$, we focus on finding $X$
within the cliques of $G-D$. As each clique type can be treated
independently, it is safe to assume for now
that all the cliques in the graph admit the same clique type.
We say that set $S \subseteq D$ is $\alpha$-\emph{bounded} (for
the clique
type under consideration) if the number of vertices with
neighborhood $S$ in $D$ is stored explicitly by the clique
type, i.e., all cliques within this clique type have the same
amount of vertices with neighborhood $S$ in $D$ and this amount
is at most $\alpha$. Moreover, we say that $S$ is
$\delta$-\emph{thin} with regard to the considered clique type and
the considered solution $X$ if for every clique $C$ within
this clique type we have
$\left| \{ v \in V(C) \mid N(v)\cap D=S \} \cap X \right| \leq \delta$, that is, the
size of the intersection of the set of vertices with neighborhood $S$ and $X$ is bounded above by the parameter $\delta$. Finally, we will say that
$S$ is $\delta$-\emph{thick} if for every clique $C$ within this clique type we have
$\left|\{ v \in V(C) \mid N(v)\cap D=S \}  \setminus X\right| \leq \delta$, that is,
the size of the intersection of the set of vertices with neighborhood $S$ and
$X$ is upper-bounded by the parameter $\delta$. We will say that a pair
$(G,X)$ has an $\alpha$-coherent shape (with regard to the considered
clique type) if each set $S \subseteq D$ is either
$\alpha$-bounded, $\alpha/2$-thin,
or $\alpha/2$-thick.
If a pair $(G,X)$ has an $\alpha$-coherent shape with regard to all
non-empty
clique types, then we just say that $(G,X)$ has an $\alpha$-coherent
shape. We illustrate the coherence property in \cref{fig:coherent}. The $\alpha$-coherence for large enough alpha is \lv{precisely} what we need to make our \FPT algorithm work.

\begin{figure}
\centering
\begin{center}
\includegraphics[width=\sv{0.75}\textwidth]{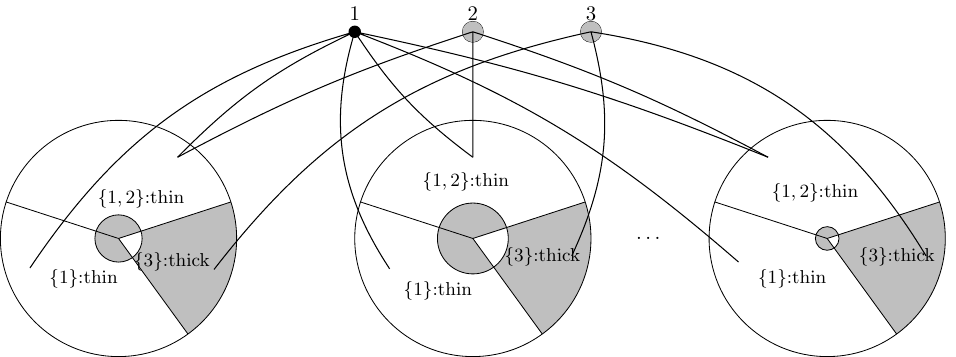}
  \end{center}

\caption{Illustration of coherence for a pair $(G,X)$.
The modulator $D$ is pictured by the three vertices at the top,
where the marked vertices belong to the solution $X$. All
cliques in $G-D$ are of the same clique type. The sets
$\{1\},\{1,2\},\{3\}$ are unbounded in this example. The
cliques are partitioned into vertices with neighborhoods
$\{1\},\{1,2\}$ and $\{3\}$ respectively (there are no
other vertices in the cliques). The solution $X$ is marked gray
within each part of each clique.
The sets $\{1\},\{1,2\}$ are thin, while the set $\{3\}$ is
thick with respect to the clique type and $X$.}\label{fig:coherent}
\end{figure}

In the following theorem, we assume the optimum modulator is given on the input as there is an \FPT algorithm parameterized by the solution size that finds it. Currently, the best bound is $1.7549^{|D_G|}\cdot \text{poly}(|V(G)|)$~\cite{FPTCVD}.

\begin{theorem}[Main informal algorithmic \FPT statement; see \cref{thm:main-fpt} for a formal version]\label{thm:main-informal}
    Let $\phi(\aleph)$ be any \MSOo formula with one free
    variable, let $G$ be a graph and let $D_G$ be a minimum
    modulator
    of $G$.
    If there exists a solution of the minimal fair cost whose
    shape is $\alpha$-coherent for sufficiently large value
    of $\alpha$ dependent purely on $|\phi|$ and $|D_G|$,
    then we can solve the respective \textsc{\MSOo-FairVE}
    problem in \FPT time parameterized by $|\phi|$ and $|D_G|$.
\end{theorem}

The assumptions of the above theorem are
properties of the problem,
and we can verify if a problem satisfies them.
We provide such verification for a variety of problems
in \cref{sec:problems}.
Here, we demonstrate this for
the Fair Vertex Cover problem.
Note that from each clique we take either all or all
but one vertices to the vertex cover.
Therefore, each solution leads to a coherent shape, as
for every clique type every subset of the modulator is $1$-thick.

We exhibit several fundamental graph vertex problems that are covered by \cref{thm:main-informal}.
We provide only very concise formulation here, for more details see \cref{sec:problems}.
Given a graph $G$ and integer $d$ in the input,
we describe the problems by expressing how $G-X$ is constrained.
Note that in addition to that we are looking for $X$ with fair cost at most $d$.
The \emph{Fair Feedback Vertex Set} problem requires $G-X$ to be a forest.
The \emph{Fair Odd Cycle Transversal} problem requires $G-X$ to be a bipartite graph.
Of course, we can impose additional \MSOo expressible constraints on $X$ in all the problems above.
In particular, the connectivity constraint on $X$ leads to variants of the problems above that have been well studied in the non-fair setting; see e.g., a recent study of those problems parameterized by clique-width~\cite{FK23}.
We also define the Fair $[\sigma,\rho]$-Domination problem for $\sigma,\rho$ finite or cofinite.
The original \emph{$[\sigma,\rho]$-Domination} problems have been coined by Telle in~\cite{T94}.
The task is to find a set of vertices $X\subseteq V(G)$ such that for all $v\in X$, $|N(v)\cap X|\in\sigma$ and for all $v\in V(G)\setminus X$, $|N(v)\cap X|\in\rho$.
An important and classical member of the family of problems above is the \emph{Dominating Set} problem, where $\sigma=\mathbb{N}, \rho = \mathbb{N}\setminus \{0\}$.
We conclude with a list of problems that are implied by \cref{thm:main-informal}.
\sv{We postpone the proof of \cref{cor:problmesFPT} to \cref{sec:problems}.%
}

\sv{\begin{restatable}[$\clubsuit$]{corollary}{CorSlvdProblems}\label{cor:problmesFPT}}
\lv{\begin{restatable}{corollary}{CorSlvdProblems}\label{cor:problmesFPT}}
  \sv{
  We can solve the following problems in \FPT time parameterized by the cluster vertex deletion:
Fair Vertex Cover,
Fair Feedback Vertex Set,
Fair Odd Cycle Transversal, %
Fair Dominating Set,
Connected variants of these problems,
and Fair $[\sigma,\rho]$-Domination problem for $\sigma$ finite, or when both $\sigma, \rho$ are cofinite.
}
  \lv{
  We can solve the following problems in \FPT time parameterized by the cluster vertex deletion:
   \begin{multicols}{2} %
\begin{itemize}
 \item
Fair Vertex Cover,
 \item
Fair Feedback Vertex Set,
\item
Fair Odd Cycle Transversal, %
   \columnbreak
 \item
Fair Dominating Set,
 \item
Connected variants of these problems,
 \item
and Fair $[\sigma,\rho]$-Domination problem for $\sigma$ finite, or when both $\sigma, \rho$ are cofinite.
\end{itemize}
\end{multicols}
}
\end{restatable}

\lv{\paragraph*{Organization of the Paper.}}
\sv{\noindent\textbf{Organization of the Paper.~~}}
We start by \lv{introducing some preliminary notation and definitions
in \Cref{sec:prelims}. 
We then move on to }%
introducing the
concept of shapes in \cref{sec:shps}. We describe the model-checking
toolbox and give proofs related to structural understanding of the
problem in \cref{sec:mc}\sv{, mostly postponed to Appendix~\ref{sec:mc-app}}.
We state the main theorem (\Cref{thm:main-fpt}) formally in \Cref{sec:FPT}, where we also
provide the algorithmic side of the proof including an ILP
formulation of the problem.
\lv{\cref{sec:problems}}\sv{Appendix~\ref{sec:problems}} provides problem statements of natural problems that are solved by our main theorem together with proofs that they are covered by it.
Finally, \cref{sec:hardness} contains the proof of our hardness result.
\lv{We conclude with short conclusions (\cref{sec:conc}).}

\sv{\noindent\textbf{Preliminaries.}\label{sec:prelims}}
\lv{\section{Preliminaries}\label{sec:prelims}}
We denote $\setupto{n}\df \{ 0,1,2,\ldots, n \}$ and $\nodd \df \{2n+1 : n \in \mathbb{N} \}$.
For a graph $\Graph$ and a set $\VSx \subseteq \Vertices(G)$, we
denote the complement of $\VSx$ as
$\NOT{\VSx}\df\Vertices(G) \setminus \VSx$, the subgraph induced by
$\VSx$ as $\GraphInd{\VSx}$, while $G-\VSx$ denotes $\GraphInd{\NOT{\VSx}}$.
For a vertex $v \in V(G)$, we denote by $N(v)$ the open
neighborhood of $v$, while $N[v]$ denotes the closed
neighborhood of $v$. For a set $\VSx \subseteq \Vertices(G)$ and a
vertex $v \in \Vertices(G)$, we denote $\nbd{\VSx}{v}\df\nb{v} \cap \VSx$, i.e., the neighbors of $v$ in~$\VSx$.
A \emph{modulator} $\MS \subseteq \Vertices(G)$ of $\Graph$ is a set
such that $G - \MS$ is a collection of
vertex-disjoint cliques.
For every graph $G$, we denote by $\MSs$ any fixed modulator of the
smallest size.
For a graph $\Graph$ and its modulator $\MSs$, let $\Cc_\Graph$
denote the collection of cliques in $G - \MSs$. The \emph{cluster
vertex deletion number} of a graph $G$, also denoted as $\cvd(G)$,
is the size of $\MSs$.
\sv{The \emph{size} of $\phi(\aleph)$ (denoted by $|\phi|$) is defined as the number of vertex and set quantifiers. }

\lv{
The formulas of \MSOo logic are those that can be constructed using vertex variables, denoted usually by $x_i, y_i, \dots$, set variables denoted usually by $X_i, Y_i, \dots$, \emph{label classes} denoted by $L_i$, the predicates $E(x_i,x_j)$, $x_i \in L_j$, and $x_i = x_j$ operating on vertex variables, standard propositional connectives, and the quantifiers $\exists, \forall$ operating on vertex and set variables. The semantics is defined in the usual way, with the predicate $E(x_i, x_j)$ being true if and only if $(x_i,x_j) \in E$ and labels being interpreted as sets of vertices.

We will also consider two more standard logics in graph theory. The logic \FO allows quantifications only over elements (vertices or edges).
A more powerful logic, \MSOt, extends \MSOo by permitting
quantification over sets of edges.
It is well known that \FO $\subsetneq$ \MSOo $\subsetneq$ \MSOt, where connectivity and the Hamiltonian cycle are properties demonstrating strict inclusions, respectively.
For more details consult \cite{CE-book}.
Let $\phi(\aleph)$ be a formula with one free variable. 
The \emph{size} of $\phi(\aleph)$ (denoted by $|\phi|$) is defined as the number of vertex and set quantifiers. 
}

\section{Shapes}\label{sec:shps}

In this section, we work towards compact representations of
solutions to a particular \MSOo-FairVE problem that we
wish to solve. \lv{In \Cref{sec:cltypes},}\sv{First,} we introduce clique types,
which compactly represent the relevant information about
cliques of the input graph.
Then, we define compliance of sets of vertices
(\cref{sec:compliant}), a property that allows us to represent
a solution set in a compact form. To justify this definition,
we show that we can assume compliance without any loss in either
the fair cost or the logical value of the formula.
After that, we give a formal definition of a shape
(\cref{sec:shape}). A shape provides a compact description of the
solution and contains all the information relevant for the
truthfulness of the formula.  %

\lv{\subsection{Clique Types}\label{sec:cltypes}}
\sv{\noindent\textbf{Clique Types.~~}\label{sec:cltypes}}
To define the clique types,
we assume from now on that a set $\MSs$ has a fixed ordering.
We denote all binary vectors of length $d$ by $\ntypes^{d}=\{0,1 \}^d$, and we call them \emph{neighborhood types}.
As we mostly use $\ntypes^{d}$ for $d=|\MSs|$, we write $\ntypes$ when it is clear from the context.
Since the vertices in $\MSs$ are ordered,
there is a natural bijection between all subsets of $\MSs$ and
vectors in $\ntypes^{|\MSs|}$.
For each $C\in\Cc_\Graph$ and for each $\nT\in\ntypes^{|\MSs|}$,
we let $C_\nT\df \{ v \in \Vertices(C) : \nbd{\MSs}{v}=\nT \}$ be
the set of vertices in $C$ whose neighborhood in $\MSs$ is $\nT$.

For $\alpha\in\nodd$\footnote{For technical ease of notation that will be apparent later, we insist on $\alpha$ being odd. At this point, this assumption is not necessary, but we prefer to make it consistent and prepare the reader from the start that $\alpha$ takes only odd values.},
 we define an \emph{$\alpha$-clique type} as vector
 \cT indexed by neighborhood types $\nT \in \ntypes$, with entries taking values from $[\alpha]$.
 We say that a clique $C \in \Cc_\Graph$ is of $\alpha$-clique
 type \cT if $\cT[\nT] = \min(|C_\nT|, \alpha)$ for each
 $\nT \in \ntypes$.
We denote as $\ctypes^{(\alpha, |\MSs|)} \df [\alpha]^{\ntypes^{|\MSs|}}$ the set of all $\alpha$-clique types.
Clearly, $|\ctypes^{(\alpha, |\MSs|)}| = (\alpha+1)^{|\ntypes^{|\MSs|}|} = (\alpha+1)^{2^{|\MSs|}}$.
We will use the shortened notation $\ctypes$ for $\ctypes^{(\alpha, |\MSs|)}$ whenever it is clear from the context.

\subsection{Compliant Sets}\label{sec:compliant}
The following definition captures which sets of vertices can be
compactly represented by our approach.

\begin{definition}[$\alpha$-compliant sets]
Let $\alpha$ be a positive integer and $G$ be a graph.
Recall that $D_G$ is the modulator of $G$ and $\Cc_G$ is the
collection of cliques in $G-D_G$.

We say that $X\subseteq V(G)$ is \emph{$\alpha$-compliant} if
for each $C \in \Cc_G$ and for each $\nT\in\ntypes$
\[\text{either } |X\cap C_\nT|\le \frac{\alpha}{2} \text{ or } |\NOT{X}\cap C_\nT|\le \frac{\alpha}{2}.\]
\end{definition}

From now on we will insist on the solution to our problem to be
$\alpha$-compliant for some parameter $\alpha$. To justify that
it is reasonable to expect $\alpha$-compliance and reveal for
which parameter $\alpha$ it is needed, we provide the lemma below.

\begin{restatable}[Compliant solution of capped fair cost]{lemma}{ComSolution}\label{lem:compliant}
    There is a computable function ${f:\mathbb{N}^2\rightarrow \mathbb{N}}$ such that for any \MSOo formula $\phi(\aleph)$ with one free variable, for any graph $G$, and for any $\alpha \in \nodd, \alpha > f(|D_G|, |\phi|)$, the following is true. For each $X\subseteq V(G)$ such that $G\models \phi(X)$, there exist $X'\subseteq V(G)$ such that:
    \nopagebreak
    \begin{multicols}{3}
        \begin{itemize}
          \item $X'$ is $\alpha$-compliant, 
          \item $G\models \phi(X')$, and
            \item $\fc{X'}\le\fc{X}$.
        \end{itemize}
    \end{multicols}
\end{restatable}

Note that $\alpha$-compliance is a weaker
property than $\alpha$-coherence introduced in
Section~\ref{sec:overview}. In particular,
$\alpha$-compliance is necessary for $\alpha$-coherence.
We postpone the proof as we first need to introduce more model-checking tools before we can prove it.

\subsection{Solution Patterns and Shapes}\label{sec:shape}
Throughout this section, we work with $\alpha$-compliant sets,
for some fixed $\alpha$ that depends only on
the formula and the size of the modulator. %

Moreover, we use
$(\alpha+1)$-clique types in order to
distinguish if for some particular clique $C \in \Cc_G$ and
$\nT \in \ntypes$ the size
of $C_\nT$ is greater than $\alpha$ or not. We say that
a pair $(\nT, \ctype)$, where $\nT\in\ntypes$ and
$\ctype\in\ctypes^{(\alpha+1, |D_G|)}$
is \emph{bounded} if $\ctype[\nT] \leq \alpha$.
Otherwise (i.e., $\ctype[\nT]=\alpha+1$), $(\nT, \ctype)$ is \emph{unbounded}.

Now we describe how an $\alpha$-compliant set $X \subseteq V(G)$
can be represented.
To do so, we introduce the concept of shapes and solution patterns. 
For some given parameters $d$ and $\alpha$, we define a \emph{solution pattern} as
a vector indexed by
neighborhood types $\nT \in \ntypes^{d}$ and with values
from $\setupto{\alpha}$. We denote all \emph{solution patterns} as
$\spatterns^{(\alpha, d)}\df[\alpha]^{\ntypes^d}$.
We use shorthand notation $\spatterns$ for
$\spatterns^{(\alpha, |D_G|)}$ when it is clear from the context.
Intuitively, every solution pattern $\sP \in \spatterns^{(\alpha, |\MSs|)}$ describes
how some potential solution $X$ is distributed over vertices
with a
specific neighborhood type $\nT \in \ntypes^{|\MSs|}$ in a
particular clique $C \in \Cc_G$.
To be more precise, if $\cT$ is the clique type of a clique
$C$,
and $\cT[\nT] \leq \alpha$, we can easily describe the exact
amount of vertices from $X$ in $C_\nT$ by a solution pattern
which can hold all possible values up to $\alpha$. If, on the
other hand, $\cT[\nT]=\alpha+1$, assuming $\alpha$-compliance
of $X$, we can describe $X \cap C_\nT$ by either storing
$\sP[\nT]=|X \cap C_\nT|$ if $|X\cap C_\nT|\leq \alpha/2$, or
by storing $\sP[\nT]=\alpha-|\NOT{X} \cap C_\nT|$ if
$|\NOT{X}\cap C_\nT| \leq \alpha/2$. In other words, the
solution pattern encodes either the cardinality of the intersection of $C_\nT$
and $X$, or it encodes the cardinality of $C_\nT \setminus X$, depending on which
of these two values is smaller than $\alpha/2$. This is
summarized in the following definition.

\begin{definition}\label{def:XmatchessP}
  For a graph $G$, $\alpha \in \nodd$ and a clique $C \in \Cc_G$ whose
  $(\alpha+1)$-clique type is $\cT$, we say that a set $X \subseteq V(G)$
  \emph{matches} a solution pattern
  $\sP \in \spatterns^{(\alpha, |\MSs|)}$ on clique $C$
  if for each $\nT \in \ntypes^{|\MSs|}$ it holds
  that \\ $|X \cap C_\nT| =
        \begin{cases*}
            \sP[\nT] & if $(\cT[\nT] \leq \alpha)$ or $(\cT[\nT] = \alpha+1$ and $\sP[\nT] < \alpha / 2)$ \\
            |C_\nT| - (\alpha - \sP[\nT]) & if $\cT[\nT] = \alpha+1$ and $\sP[\nT] > \alpha/2$
          \end{cases*}$
\end{definition}

The above definition naturally partitions the triplets
$(\nT, \ctype, \sP)$ where $\nT\in\ntypes$,
$\ctype\in\ctypes^{(\alpha+1, |D_G|)}$, and
$\sP\in\spatterns^{(\alpha, |D_G|)}$ into the following types:
  \begin{itemize}
    \item[{\bf bounded:}] $(\nT, \ctype, \sP)$ is \emph{bounded} if
      $(\nT,\ctype)$ is bounded (otherwise $(\nT, \ctype, \sP)$ is \emph{unbounded})
    \item[{\bf thin:}] $(\nT, \ctype, \sP)$ is \emph{thin}
           if it is unbounded and $\sP[\nT]<\alpha/2$,
    \item[{\bf thick:}] $(\nT, \ctype, \sP)$ is \emph{thick}
           if it is unbounded and $\sP[\nT] > \alpha/2$.
  \end{itemize}
  Note that as $\alpha$ is odd $(\nT, \ctype, \sP)$ is either thin, thick, or bounded.

\medskip

Having defined solution patterns and their meaning, we move on to
defining the most important combinatorial object of the paper.

\begin{definition}[Shapes]\label{def:shape}
  Let $\alpha \in \nodd$ and $\gamma, d \in \mathbb{N}$. 
  We define $\shapes^{(\alpha,\gamma, d)}$ as a set of all tuples $(\nT^*,\mtx)$, where:
  \begin{enumerate}
    \item[1)] $\nT^* \in \ntypes^d$, %
    \item[2)] $\mtx$ is a $|\spatterns^{(\alpha, d)}| \times |\ctypes^{(\alpha+1, d)}|$ matrix 
  whose entries are from $\setupto{\gamma}$, and %
\item[3)] \phantomsection\label{c:3} for each $\nT\in\ntypes$, $\cT\in\ctypes^{(\alpha+1, d)}$, and $\sP\in\spatterns^{(\alpha, d)}$ such that $\mtx[\sP, \cT]>0$  and  $\cT[\nT] \leq \alpha$, we have $\sP[\nT]\le\cT[\nT]$.
\end{enumerate}
We refer to an element $\shape\in\shapes^{(\alpha,\gamma,d)}$ as a \emph{shape}.
\end{definition}

\medskip

We note that Condition~\hyperref[c:3]{3)} in \cref{def:shape} is
necessary, as it ensures that whenever the triplet $(\nT,\cT,\sP)$
is bounded,
the solution pattern $\sP$ does not describe a solution containing
more vertices than are present in a corresponding subset $C_\nT$
of a clique $C$ of type $\cT$.

We now provide some intuition behind the shapes. First, we briefly
describe the meaning of the parameters $\alpha, \gamma$. Roughly
speaking, $\alpha$ corresponds to the maximum number of vertices
of the same neighborhood within a single clique that the formula
can distinguish. In other words, once there are more than $\alpha$
such vertices, increasing their number further makes no difference
to the formula. On the other hand, $\gamma$ corresponds to the
maximum number of cliques of the same $(\alpha+1)$-clique type the
formula can distinguish. We will formalize this intuition
in \cref{prop:lam} and \cref{cor:induction}. Thus,
we restrict ourselves to shapes with values bounded by these
parameters.
In other words, given $G$ and $\phi(\aleph)$, the shapes in
$\shapes^{(\alpha, \gamma, |\MSs|)}$ describe all
the information about the solution that is relevant for the formula,
as $\alpha$ and $\gamma$ are computed from $\phi(\aleph)$ and $\MSs$.
Let us now intuitively describe how the shapes are interpreted further.
Each such shape, denoted as \shape, consists of two parts that
are supposed to guess some potential solution $X$.
First, the shape guesses the portion of the potential solution
set $X$ in the
modulator $\MSs$. More precisely, $X \cap \MSs$ forms the first part
of $\shape$, denoted $\nT^*$. The second part is a matrix, that
guesses how the solution $X$ is distributed over the cliques and
their respective neighborhood parts. More precisely,
the entry $\mtx[\sP,\cT]$ represents the minimum of two values:
parameter $\gamma$ and the number of cliques $C$ whose
$(\alpha+1)$-clique
type is $\cT$ and such that $X$ matches $\sP$ on $C$.
Notice that the number of all possible shapes is bounded in terms of $\alpha, \gamma, |\MSs|$, so we can enumerate all of them, and maintain an \FPT complexity.

\begin{definition}[The shape of a graph and a set]
For a graph $G$ and an $\alpha$-compliant set
$X \subseteq V(G)$ we define a shape of $X$ in $G$,
denoted as $\shapefn^{(\alpha,\gamma)}(G,X) \in \shapes^{(\alpha,\gamma, |D_G|)}$,
as follows:
\begin{itemize}
 \item $\nT^*=X \cap D_G$
 \item $M[\sP,\cT]=\min (\gamma, |\{ C \in \Cc_G: C \text{ is
 of type } \cT \text{ and } X \text{ matches } \sP \text{ on } C \}|)$\\ for all
 $\sP \in \spatterns^{(\alpha,|D_G|)}$ and $\cT \in \ctypes^{(\alpha+1,|D_G|)}$
\end{itemize}

\end{definition}

We already described how to assign a shape to a pair $(G,X)$.
It is also possible, given a shape $\shape$, to produce a
canonical pair $(G',X')$ whose shape is $\shape$.

\medskip

\begin{definition}[The graph and the set associated with a shape]\label{def:associatedGraph}
  Given $(\nT^*, \mtx)=\shape \in \shapes^{(\alpha,\gamma, d)}$,
  we define the graph $G'$ and the set $X'\subseteq V(G')$
  \emph{associated} with $\shape$ as follows:
  \begin{itemize}
    \item[1)] $G'$ has $d$ vertices in the modulator, and
      exactly those represented by vector $\nT^*$ are in $X'$;
    \item[2)] For each $\cT \in \ctypes^{(\alpha+1, d)}$, $G'$ contains exactly
  $\sum_{\sP \in \spatterns^{(\alpha, d)}} \mtx[\sP,\cT]$
  cliques of clique type $\cT$, as described by the column indexed by $\cT$
  (each clique contains exactly $\cT[\nT]$ vertices of neighborhood type $\nT$
  for each $\nT\in\ntypes$).
  Moreover, for each $\sP\in\spatterns^{(\alpha, d)}$, the set $X'$ matches the
  solution pattern $\sP$ on exactly $\mtx[\sP,\cT]$ cliques of clique type $\cT$.
  \end{itemize}
\end{definition}

\medskip

We established that a shape can be interpreted as
a pair $(G',X')$, which is a canonical instance
corresponding to this shape. This allows us to evaluate a formula
on $G'$, where $X'$ is substituted for the free variable.
We can view this process as evaluating the formula on the shape
directly. This is formalized in the following definition.

\begin{definition}[Evaluation of a shape]\label{def:shapeEvaluation}
  Let $\phi(\aleph)$ be a fixed formula with one free variable.
  For each $\shape\in\shapes^{(\alpha,\gamma, d)}$, we define whether $\shape$ is \emph{true} or \emph{false} for the formula $\phi(\aleph)$ (denoted as $\models \phi(\shape)$).
  Let $G',X'$ be the graph and the set associated with $\shape$.
  Then $\models\phi(\shape)$ if and only if $G'\models\phi(X')$.
\end{definition}

We proceed with the definition of coherent shapes.
Our final \FPT algorithm presented in Section~\ref{sec:FPT} is
only able to find $\alpha$-compliant solutions whose shape is
coherent.

\begin{definition}[Coherent shape]\label{def:coherent}
  Let $\alpha \in \nodd$ and $\gamma, d \in \mathbb{N}$.
  We say that $(\nT^*,\mtx)=\shape\in\shapes^{(\alpha,\gamma, d)}$ is \emph{coherent} if, for each $\nT\in\ntypes$ and $\cT\in\ctypes^{(\alpha+1, d)}$ such that $(\nT,\cT)$ is unbounded, the elements of the set
  \[
    \{(\nT,\cT,\sP)\mid \sP\in\spatterns^{(\alpha, d)} \text{ and } \mtx[\sP,\cT] > 0 \}
  \]
  are either all thin or all thick. In this case, we say that $(\nT, \cT)$ is \emph{thin} or \emph{thick}, respectively.
\end{definition}

\lv{\section{Main Tools and Model-checking Machinery}\label{sec:mc}}
\sv{\subsection{Main Tools and Model-checking Machinery}\label{sec:mc}}

The purpose of this section is to confirm that the truthfulness
of the formula on an $\alpha$-compliant set is equivalent to
the truthfulness on its shape.  %

\begin{restatable}{lemma}{ShapeEquivalence}\label{lem:shape_equivalence_one}
  There is a computable function ${f:\mathbb{N}^2\rightarrow \mathbb{N}}$ such that, for any \MSOo formula $\phi(\aleph)$ with one free variable, for any graph $G$, and for any $\alpha \in \nodd$ and $\gamma \in \mathbb{N}$ such that $\alpha,\gamma > f(|D_G|, |\phi|)$, the following holds.
  If $X\subseteq V(G)$ is $\alpha$-compliant
  then $\models \phi(\shapefn^{(\alpha,\gamma)}(G,X))$ if and only if $G \models \phi(X)$.
\end{restatable}
\sv{Due to space constraints this section is postponed to Appendix~\ref{sec:mc-app}.}

\lv{

Before we prove \cref{lem:shape_equivalence_one}
in \cref{sec:mc:mainProofs} we need to introduce model-checking tools and notation.
In the remainder of this section, we introduce known logical tools
useful for proving Lemma~\ref{lem:shape_equivalence_one}, which is
crucial for the \FPT algorithm in \Cref{sec:FPT}.
}

\toappendix{
We will consider labelled graphs,
which are more convenient for proving some properties of the formula
on a graph. Recall that a label is simply a subset of vertices.

\begin{definition}
  For a labelled graph $G$, we say that vertices $u$ and $v$ have the
  same labelled type if they have the same labels and they have the same
  closed neighborhoods, i.e., $N[u]=N[v]$.
\end{definition}

We will follow the approach of irrelevant vertices and irrelevant cliques, as it was done in \cite{KMT19} and \cite{Lampis12}.

\begin{definition}\label{def:formula_to_sentene}
    For a given labelled graph $G$, a set $X \subseteq V(G)$, and a
    formula $\phi(\aleph)$, we define a sentence $\phi_X$ in the
    following way.
    We introduce a new label $\ell_X$ and assign $\ell_X$ to all vertices in
    $X$, thus obtaining the labelled graph $G_X$. We create the sentence
    $\phi_X$ from $\phi(\aleph)$ by replacing every occurrence of
    $\aleph$ with $\ell_X$.
\end{definition}
We observe that $\phi_X$ is a valid expression, since $\ell_X$ represents a set of vertices, as $X$ does.

\begin{observation}\label{upd}
  Let $G$ be a labelled graph, $X \subseteq V(G)$, and
  $\phi(\aleph)$ be a formula. Then $G \models \phi(X) \iff G_X \models \phi_X.$
\end{observation}
We now state a theorem that will be useful for proving \cref{prop:lam},
which is a crucial tool that allows us to maintain the truthfulness of
the formula under graph trimming operations.

\begin{theorem}[Reformulation of {\cite[Lemma 5]{Lampis12}}]\label{ref:lam}
   Let $G$ be a labelled graph and $\phi$ be an \MSOo sentence with $q_S$ set quantifiers and $q_v$ vertex quantifiers. Let $T \subseteq V(G)$ consist of vertices of the same labelled type. Also let $Q$ be any subset of $T$ such that $|T \setminus Q| \geq 2^{q_S}q_v$. Then $G \models \phi$ if and only if $G - Q \models \phi$.
\end{theorem}
Theorem~\ref{ref:lam} allows us to trim the graph without
changing the logical value of the formula, given that there
are many vertices of the same labelled type. For our
purpose, we want to
be able to trim graphs and their subsets of vertices without
changing the logical value of a formula on the particular
subset. This can be achieved with the trimming set defined below.
\begin{definition}[Trimming set]\label{def:trimmingSet}
Let $\tau \in \Nn$. Let $G$ be a labelled graph,
$X \subseteq V(G)$ be any subset of vertices and
$T \subseteq V(G)$ be a subset of vertices of the same
labelled type such that $|T| \geq \tau$.
We define a \emph{$\tau$-trimming set} $Q$ of $(G,X,T)$ depending
on $|X \cap T|$ as follows.
  \newline \textbf{Case 1:}  $|X\cap T| \leq \frac{\tau}{2}$. We set $Q$ to be any subset of $T \setminus X$ of size $|T| - \tau$.
  \newline \textbf{Case 2:} ${\frac{\tau}{2} < |X\cap T| \leq |T| - \frac{\tau}{2}}$.
    Let $Q_1$ be any subset of $X\cap T$ of size $|X \cap T|-\frac{\tau}{2}$ and let $Q_2$ be any subset of $T \setminus X$ of size $|T \setminus X|-\frac{\tau}{2}$.
    Let $Q \coloneqq Q_1 \dot\cup Q_2$.
 \newline    \textbf{Case 3:} $|T| - \frac{\tau}{2} < |X\cap T|$.
    Let $Q$ be any subset of $X\cap T$ of size $|T| - \tau$.
\end{definition}

The meaning of the trimming sets is explained by the
next corollary.

\begin{restatable}[Corollary of {\cite[Lemma 5]{Lampis12}}]{corollary}{CorLampis}\label{prop:lam}
    Let $\phi(\aleph)$ be any \MSOo formula with one free variable, $q_v$ vertex quantifiers, and $q_S$ set quantifiers.
    Let $\tau \geq 2 \cdot 2^{q_S}q_v$.
    Let $G$ be a labelled graph and $T \subseteq V(G)$ be a set of
    vertices of the same labelled type, such that
    $|T| \geq \tau$. Let $X \subseteq V(G)$.
    Then the $\tau$-trimming set $Q$ of $(G,X,T)$ satisfies:
    \begin{align}
      G \models \phi(X)  &\iff %
        G - Q \models \phi(X \setminus Q) 
            \label{eq:TwinsReduction-case1}\\
              G - Q \models \phi(X \setminus Q) &\implies 
            \begin{cases}   
              G \models \phi( X \setminus Q)  & \text{if } |X\cap T| \leq |T| - \frac{\tau}{2}, \\
              G \models \phi( X \cup Q)   & \text{if } |X \cap T| \geq \frac{\tau}{2}.
            \end{cases} \label{eq:TwinsReduction-case2} 
    \end{align}   
    \begin{align}
      \begin{cases}   \label{eq:XcapQ}
              |(T \setminus Q)\cap X|= |X\cap T| &  \text{if } |X\cap T| \leq \frac{\tau}{2}, \\
              |(T \setminus Q)\cap X|=\lfloor\frac{\tau}{2}\rfloor &  \text{if } \frac{\tau}{2}< |X\cap T| \leq |T| - \frac{\tau}{2}, \\
               |(T \setminus Q)\cap X|=\tau-|T\setminus X|  & \text{otherwise.}
            \end{cases}
          \end{align}%
          Moreover,
          $|T \setminus Q|=\tau$.
\end{restatable}
Observe that the above corollary allows us to reduce
the number of vertices of the same labelled type without
altering the logical value of the formula, provided
that there are sufficiently many such vertices.
  \sv{ \CorLampis* \begin{proof}[Proof of \cref{prop:lam}]}
    \lv{\begin{proof}}
  Let $Q$ be a $\tau$-trimming set of $(G,X,T)$.
Observe that $Q$ is defined so that $|T\setminus Q|=\tau$.
    Also, we can now easily check that $Q$ satisfies \cref{eq:XcapQ} in all three cases.

  We now prove \Cref{eq:TwinsReduction-case1}. 
  Given $X$, we can use \Cref{upd}. 
  In \textbf{Case 1}, we have $Q\subseteq T\setminus X$ and $|(T\setminus Q)\setminus X|\ge \frac{\tau}{2}$; hence, after labeling $X$, the vertices in $Q$ have the same labelled type as at least $\frac{\tau}{2}$ vertices of $T\setminus Q$. Therefore, we can directly apply \cref{ref:lam}.
  In \textbf{Case 3}, we have $Q\subseteq X\cap T$ and $|(T\setminus Q)\cap X|\ge \frac{\tau}{2}$; hence, after labeling $X$, the vertices in $Q$ have the same labelled type as at least $\frac{\tau}{2}$ vertices of $T\setminus Q$. Therefore, we can directly apply \cref{ref:lam}.
  In \textbf{Case 2}, we apply \cref{ref:lam} twice: first on $Q_1$ and then on $Q_2$.
  Still, after labeling $X$, in each such step, $Q_1$ (respectively $Q_2$) has the same labelled type as at least $\frac{\tau}{2}$ vertices of $T\setminus Q$ with the corresponding label status.
  Therefore, in all cases, we have
    \[G \models \phi(X)  \iff  G_X - Q \models \phi_X. \]
   Observe that $G_X - Q \models \phi_X$ is equivalent to $G_{X \setminus Q} - Q \models \phi_{X \setminus Q}$.
   Hence, we conclude the proof of \Cref{eq:TwinsReduction-case1} by applying \Cref{upd}.

   We conclude the proof by proving \Cref{eq:TwinsReduction-case2}. 
We consider two cases depending on the size of the set $X\cap T$.
Suppose $|X\cap T| \leq |T| - \frac{\tau}{2}$.
By construction of $Q$, we know that $T\setminus Q$ contains at least $\frac{\tau}{2}$ vertices \textbf{not} in $X$.
Assume $G-Q \models \phi(X\setminus Q)$.
Using \Cref{upd}, we get that this is equivalent to $G_{X\setminus Q}-Q \models \phi_{X\setminus Q}$.
In turn, there are at least $\frac{\tau}{2}$ vertices of the same labelled type \textbf{not} labelled by $X$, and we use this labeling on additional $|Q|$ vertices, which we add to $G-Q$, before we apply \Cref{ref:lam} to it.
Hence, we obtain $G_{X\setminus Q}-Q \models \phi_{X\setminus Q} \iff  G_{X\setminus Q} \models \phi_{X\setminus Q} 
$ as we know that the newly added vertices are not labelled by $X$.
We conclude by \cref{upd}.
Now, suppose $|X \cap T| \geq \frac{\tau}{2}$.
By construction of $Q$, we know that $T\setminus Q$ contains at least $\frac{\tau}{2}$ vertices in $X$.
Following the same steps as before, there are at least $\frac{\tau}{2}$ vertices of the same labelled type labelled by $X$, and we use this labeling on additional $|Q|$ vertices, which we add to $G-Q$, before we apply \Cref{ref:lam} to it.
Hence, we obtain 
$G-Q \models \phi(X\setminus Q) \iff  G_{X\cup Q} \models \phi_{X\cup Q}$ as we know that the newly added vertices are labelled by $X$.
We conclude by \cref{upd}.
\end{proof}

We now move on to developing tools that allow us to remove cliques from
the graph while maintaining the truthfulness of the formula.
\begin{definition}
    We say that two cliques have the same labelled clique type if there is a bijection between their vertices that preserves the labelled type.
\end{definition}

We now state an analogue of the irrelevant clique lemma for bounded twin cover \cite[Lemma 7\footnote{Lemma 8 in the full version of the paper \cite{KMT19-arxiv}}]{KMT19}. 
We decided to omit the proof, as literally the same proof without any changes works in our case as well.

\begin{lemma}[Corollary of the proof of {\cite[Lemma 8]{KMT19-arxiv}}]\label{lem:induction}
Let $G$ be a labelled graph with the minimum modulator $D_G$.
Let $\phi$ be an \MSOo sentence with
$q_v$ vertex quantifiers and $q_S$ set quantifiers. Suppose the size
of a maximum clique in $G - D_G$ is bounded by $\ell$. If there are strictly more than
$\gamma(q_S, q_v, \ell) \coloneqq 2^{\ell \cdot q_S}(q_v + 1)$
cliques of the same labelled clique type $T$, then there exists a clique $C$ of labelled clique type $T$ such
that $G \models \phi$ if and only if $G - C \models \phi$.
\end{lemma}

The above lemma, which holds for sentences, effectively allows us to shrink the graph so that it does not contain too many similar cliques. 
Now, we want to prove a similar statement for \MSOo formulas with one free variable.

\begin{corollary}\label{cor:induction}
  Let $G$ be a labelled graph with the size of a maximum clique in
  $\Cc_G$ bounded by~$\ell$.
  Let $X\subseteq V(G)$, and let $\phi(\aleph)$ be any \MSOo formula with $q_v$ vertex quantifiers and $q_S$ set quantifiers. If there are strictly more than
  \sv{$\gamma(q_S,q_v, \ell) \coloneqq 2^\ell \cdot 2^{\ell \cdot q_S}(q_v+1)$}\lv{\[\gamma(q_S,q_v, \ell) \coloneqq 2^\ell \cdot 2^{\ell \cdot q_S}(q_v+1)\]} cliques of the same labelled clique type $T$, then there exists a clique $C$ of labelled clique type $T$ such that $G\models \phi(X)$ if and only if $G - C \models \phi(X \setminus C)$.
\end{corollary}

\begin{proof} %
  Let $G_X$ and $\phi_X$ be the labelled graph and the logic sentence
  as defined in \cref{def:formula_to_sentene}. It suffices to prove that
  $G_X \models \phi_X \iff G_X - C \models \phi_X$.
  Observe that in the graph $G_X$ there exist strictly more than
  $2^{\ell \cdot q_S}(q_v+1)$ cliques of the same labelled clique type:
  indeed, in a clique of size at most $\ell$, the label $\ell_X$ can appear
  in at most $2^\ell$ different ways, so among the strictly more than
  $2^\ell \cdot 2^{\ell \cdot q_S}(q_v+1)$ cliques of clique type $T$ in $G$,
  at least one labelled clique type appears strictly more than
  $2^{\ell \cdot q_S}(q_v+1)$ times in $G_X$.
  The claim now follows from \cref{lem:induction}.
\end{proof}
}%

\lv{With the above tools, we are ready to prove Lemma~\ref{lem:shape_equivalence_one}.}

\toappendix{

\ShapeEquivalence*

\begin{proof}[Proof of Lemma~\ref{lem:shape_equivalence_one}]
  To prove the lemma we first set
  $f(|D_G|,|\phi|)=2^\ell \cdot 2^{\ell \cdot q_S} \cdot (q_v+1)$
  for $\ell=2 \cdot 2^{q_S}q_v \cdot 2^{|D_G|}$.%
  Let $G$ be a graph %
and let $X \subseteq V(G)$ be an $\alpha$-compliant set as in the
statement of the lemma.
Let $(G',X')$ be the canonical
pair associated with shape $\shapefn^{(\alpha,\gamma)}(G,X)$.
We need to prove that $G' \models \phi(X') \iff G \models \phi(X)$.

Let us now consider a trimming process applied to $(G,X)$, defined as
follows. In Phase~$1$, while there is a
clique $C \in \Cc_G$ with
$|C_\nT| > \alpha+1$ for some $\nT \in \ntypes^{(|D_G|)}$,
we find an $(\alpha+1)$-trimming set
$Q$ for $(G,X,C_\nT)$ according to Definition~\ref{def:trimmingSet}. For that, we interpret $G$ as a labelled graph with an empty set of labels. We then reduce $(G,X)$ to
$(G - Q,X \setminus Q)$ without changing the logical
value of the formula, based on Corollary~\ref{prop:lam}.
In Phase $2$, if there are more than $\gamma$ cliques of
some type $\cT$ on which $X$ matches some solution pattern
$\sP$, then we can iteratively reduce their number to $\gamma$ while preserving the logical value of the
formula using Corollary~\ref{cor:induction}. Let now
$(G'',X'')$ be the pair obtained from $(G,X)$ after
the trimming process completes. By Corollary~\ref{prop:lam} and Corollary~\ref{cor:induction}, we get that $G \models \phi(X) \iff G'' \models \phi(X'')$. To complete the proof
it is sufficient to show that $(G',X')$ is isomorphic to
$(G'',X'')$.

So let $C \in \Cc_G$ be some clique of $G-D_G$ and let
us assume $C$ was not removed in Phase $2$. Let us first consider how the trimming process affects $C$ and
$X \cap C$.
Let $\nT$ be such that the trimming applies to $C_\nT$, i.e.,
$|C_\nT| > \alpha+1$.
Since $X$ is $\alpha$-compliant,
$$\begin{cases}
|(C_\nT \setminus Q) \cap X| = |X \cap C_\nT| \text{\;\;\;\;\;\;\;\;\;\; \;\;\;\;\;    if }
 |X \cap C_\nT| < \alpha/2, \text{ and }\\
|(C_\nT \setminus Q) \cap X| =\alpha+1-(|C_\nT \setminus X|)  \text{\; otherwise. }
\end{cases}$$
Moreover, if $C_\nT''$ is the set $C_\nT$ after the
trimming process completes, then $|C_\nT''|=|C_\nT \setminus Q|=\alpha+1$.

Now, let us consider how the clique $C$ contributes to
$\shapefn^{(\alpha,\gamma)}(G,X)$.
Let $\cT$ be an $(\alpha+1)$-clique type of $C$.
Let
$\sP\in\spatterns^{(\alpha,|D_G|)}$ be a solution pattern
such that $X$ matches $\sP$ on $C$. Thus, in $\shapefn^{(\alpha,\gamma)}(G,X)$, we get
$M[\sP,\cT] > 0$ and $C$ contributes towards $M[\sP,\cT]$.
Hence, in $(G',X')$, there is a clique $C'$ contributing
towards $M[\sP,\cT]$. Let us consider how the solution
$X'$ is distributed within $C'$.

Let us consider $\nT$ such that $|C_\nT| \geq \alpha+1$;
then
we have $\cT[\nT]=\alpha+1$.
By Definition~\ref{def:XmatchessP},
$$\begin{cases}
\sP[\nT]=|C_\nT \cap X| \text{\;\;\;\;\;\;\;\; if }
|X \cap C_\nT|<\alpha/2, \text{ and }\\
\sP[\nT]=\alpha-|\NOT{X}\cap C_\nT| \text{\; otherwise.}
\end{cases}$$
By Definition~\ref{def:associatedGraph}, $|C'_\nT|=\alpha+1$ and
$$\begin{cases}
|X' \cap C'_\nT|=|C_\nT \cap X| \text{\;\;\;\;\;\;\;\;\;\;\;\;\;\;\;\;\;\;\;\;\;\; \; \; \; \; \; \; \; \; \; \; \; if }
|X \cap C_\nT|<\alpha/2, \text{ and }\\
|X' \cap C'_\nT|=|C'_\nT|-\alpha+\sP[\nT]=
\alpha+1-|C_\nT \setminus X| \text{\; otherwise.}
\end{cases}$$
We conclude that $|C'_\nT|=|C_\nT''|$ and
$|X' \cap C'_\nT|=|X \cap (C_\nT \setminus Q)|=|X'' \cap C_\nT''|$.

Let us now
consider the case when $\nT$ is such that
$|C_\nT| \leq \alpha$. In that case the trimming does not
affect $C_\nT$: $C_\nT''=C_\nT$ and $X'' \cap C_\nT''=X \cap C_\nT$.
Also, $\cT[\nT]=|C_\nT|$ and $\sP[\nT]=|X \cap C_\nT|$. Thus,
$|C'_\nT|=|C_\nT|$ and $|C'_\nT \cap X'|=|C_\nT \cap X|$.
We again conclude that $|C'_\nT|=|C_\nT''|$ and
$|X' \cap C'_\nT|=|X \cap C_\nT|=|X'' \cap C_\nT''|$.

Let us now discuss the situation when the considered clique
$C$ is removed in Phase $2$ of the trimming process. In that
case there are $\gamma$ cliques $C''_1 \ldots C''_\gamma$
isomorphic to $C''$ in $G''$ with $X''$ distributed in the
exact
same way. Similarly, there are $\gamma$ cliques
$C'_1 \ldots C'_\gamma$ isomorphic to $C'$ with $X'$
distributed in the
exact
same way. $C'$ is in turn
isomorphic to $C''$, and $X' \cap C'$ is distributed in the
same way as $X'' \cap C''$.
Consequently, after the trimming process completes, we obtain
a pair
$(G'',X'')$ isomorphic to $(G',X')$. This proves the claim.
\end{proof}

\lv{
\subsection{Proofs of Main Model-checking Statements}\label{sec:mc:mainProofs}

Before we prove \cref{lem:compliant} and \cref{lem:shape_equivalence_one}, we formulate the model-checking corollary.
}%
\sv{We now formulate the model-checking corollary.}
\begin{corollary}\label{cor:mc}
  \MSOo model-checking can be done in \FPT time parameterized by the cluster vertex deletion and size of the formula.
\end{corollary}

\begin{proof}[Proof sketch]
     Let $q_S, q_v$ be the numbers of set quantifiers and vertex quantifiers in $\varphi$, respectively, and let $d \coloneqq |\MSs|$, i.e., the size of the modulator in $G$.
     We specify $\alpha \coloneqq 2\cdot 2^{q_S}q_v$,
     $\gamma \coloneqq 2^{2^d\alpha}\cdot 2^{2^d\alpha \cdot q_S}(q_v+1)$ (recall that $\alpha,\gamma$ are chosen as in \cref{prop:lam} and \cref{cor:induction}).
     \Cref{lem:compliant} states that we can find the optimal solution even when restricted to $\alpha$-compliant sets.
     Therefore, we iterate over all $\shape\in\shapes^{(\alpha,\gamma,d)}$, whose number is bounded by the parameters.
     On each $\shape$ we model-check the formula on $\shape$, which leads to model-checking on the associated graph and set (consult \cref{def:associatedGraph}), whose sizes are bounded by the parameters.
     We conclude by \cref{lem:shape_equivalence_one}, which establishes correctness of the approach above, as $\varphi(\shape)$ is true if and only if there is an $\alpha$-compliant set $X$ such that $G\models \varphi(X)$.
\end{proof}

We note that the above approach to model-checking is not the simplest;
however, it illustrates the overall approach to model-checking using the shapes and other tools we build in order to be able to work towards the fair problems.
We conclude the section by proving \Cref{lem:compliant}.

\lv{ \ComSolution*}
\begin{proof}[Proof of \cref{lem:compliant}]

Let $f(|\MSs|, |\phi|) \coloneqq 2 \cdot 2^{q_S}q_v$ and fix any odd $\alpha \geq f(|\MSs|, |\phi|)$. The statement is trivial if $X$ is $\alpha$-compliant on all $C_\nT$ where $C\in \Cc$ and $\nT \in \ntypes$. 
Assume that there exists $C_\nT$ on which $X$ is not $\alpha$-compliant. Observe that
$\frac{\alpha}{2} < |X \cap C_\nT| < |C_\nT| - \frac{\alpha}{2}$.
Let $Q$ be a \emph{$\tau$-trimming set} of $(G,X,C_\nT)$ of size $|C_\nT| - \lfloor\frac{\alpha}{2}\rfloor$, as defined in \cref{def:trimmingSet}. 
    
Thanks to \cref{eq:TwinsReduction-case1} we know that $G - Q \models \phi(X \setminus Q)$. Then we can apply the first case of \cref{eq:TwinsReduction-case2} to obtain $G \models \phi(X \setminus Q)$. According to \cref{eq:XcapQ} we get that $|(C_\nT \setminus Q) \cap X| = \lfloor\frac{\alpha}{2}\rfloor$, and trivially
\[
|(C_\nT \setminus Q) \cap X| = |(X \setminus Q) \cap C_\nT| = \lfloor\frac{\alpha}{2}\rfloor.
\]
Let $X' \coloneqq X \setminus Q$ and observe that $X' \subseteq X$ and $G \models \phi(X')$. We can apply the previous procedure to every $C_\nT$, each time reducing the set $X'$. Obviously, after the last reduction the set $X'$ is $\alpha$-compliant. The other two properties are also satisfied because $X' \subseteq X$ and we maintain that $G \models \phi(X')$.
\end{proof}

}%

\section{\FPT Algorithm}\label{sec:FPT}
The main result of this section is an \FPT algorithm for some \MSOo-\FairVE problems parameterized by $|D_G|$ and $|\phi|$, which proves \cref{thm:main-fpt}.
Recall that we cannot give an \FPT algorithm for the most general version of the problem, as we show in \cref{sec:hardness} that the problem is \W{1}-hard. 
Hence, here, with the conditions we impose on the problem instance,
we are close to the fine line where the \FairVE problems become hard.
Now, we state our main theorem.

\begin{restatable}[FPT algorithm]{theorem}{fptThm}\label{thm:main-fpt}
     There is a computable function ${f:\mathbb{N}^2\rightarrow \mathbb{N}}$ such that for any \MSOo formula $\phi(\aleph)$ with one free variable, for any graph $G$,
     and any $\alpha \in \nodd$ and $\gamma \in \mathbb{N}$ such that $\alpha, \gamma > f(|D_G|, |\phi|)$, the following is true:

     Let $\Xx$ be the collection of all $\alpha$-compliant sets
     $X\subseteq V(G)$ such that $G\models \phi(X)$.
     If there exists $X\in\Xx$ such that
     $\fc{X}\le \min_{X'\in\Xx} \fc{X'}$ and such that $\shapefn^{(\alpha,\gamma)}(G,X)$ is coherent,
     then we can solve the \textsc{\MSOo-FairVE} problem in
     \FPT time parameterized by $|\phi|$, $\alpha$, $\gamma$,
     and $|D_G|$.
     \end{restatable}

Now we will provide some intuition on how to
prove \cref{thm:main-fpt}.
Due to \cref{lem:shape_equivalence_one}, the logical value
of $G \models \phi(X)$ depends purely on $\shapefn^{(\alpha,\gamma)}(G,X)$.
To conclude, observe that to prove \cref{thm:main-fpt}, it
is enough
to find a set $X \subseteq V(G)$ of minimum fair cost
such that
$\shape=\shapefn^{(\alpha,\gamma)}(G,X)$ for a coherent 
$\shape \in \shapes^{(\alpha, \gamma, |\MSs|)}$.  To do so,
we formulate an ILP instance whose constraints encode the required fair cost.
We first introduce a function $\cliFC()$ used to encode
the fair cost contribution of vertices in the cliques.
\begin{definition}\label{def:cfc}
  For every graph $G$, every clique $C \in \Cc_G$, any $\nT^* \subseteq \MSs$, and every solution pattern $\sP \in \spatterns^{(\alpha, |\MSs|)}$, we define the fair cost of the clique $C$ as follows, where $\cT$ is the clique type of $C$:
\begin{align*}
  \cliFC(C, \sP, \nT^*) &= \sum_{\substack{\nT \in \ntypes :\\ (\sP[\nT] < \alpha / 2 \\ \text{ or } \cT[\nT] \leq \alpha)}} \sP[\nT]
  +\sum_{\substack{\nT \in \ntypes :\\ (\cT[\nT]=\alpha+1 \\ \text{ and } \sP[\nT] > \alpha / 2)}} \left( |C_\nT| - \alpha + \sP[\nT]\right)  \\
                        &\quad+\max_{\nT \in \ntypes : C_\nT \neq \emptyset}\Delta(\nT), \text{ where}
\end{align*}
\[
\Delta(\nT)=
\begin{cases}
  |\nT \cap \nT^*|-1 & \text{if } \sP[\nT]=\alpha \text{ or } \sP[\nT]=\cT[\nT], \\
  |\nT \cap \nT^*|   & \text{otherwise.}
\end{cases}
\]
\end{definition}

\begin{observation}[\sv{$\clubsuit$}]\label{obs}
Let $G$ be a graph, $C \in \Cc_G$, $\nT^* \subseteq \MSs$, and $\sP \in \spatterns^{(\alpha, |\MSs|)}$, and let $\cT$ be the clique type of $C$.
It holds that $\cliFC(C, \sP, \nT^*) = \max_{v \in C} |N(v)\cap X|$, assuming that the solution $X$ is distributed on the clique $C$ and the modulator $\MSs$ according to $\sP$ and $\nT^*$, respectively.    
\end{observation}
\toappendix{
  \sv{\section{Missing proof from \cref{sec:FPT}}}
  \lv{\begin{proof}}
    \sv{\begin{proof}[Proof of \cref{obs}]}
Consider $\nT \subseteq D_G$. Let $X$ be a solution matching $\sP$ on $C$ and such that $X \cap D_G = \nT^*$. If $\sP[\nT] < \alpha$ and $\sP[\nT] \neq \cT[\nT]$, then there is a vertex $v \in \C_{\nT} \setminus X$. Such vertex $v$ is a neighbor of $$\sum_{\nT' \in \ntypes} |X \cap C_{nT'}|+ |\nT^* \cap \nT|=$$ $$ \sum_{\substack{\nT \in \ntypes :\\ (\sP[\nT] < \alpha / 2 \\ \text{ or } \cT[\nT] \leq \alpha)}} \sP[\nT]
  +\sum_{\substack{\nT \in \ntypes :\\ (\cT[\nT]=\alpha+1 \text{ and } \\  \sP[\nT] > \alpha / 2)}} \left( |C_\nT| - \alpha + \sP[\nT]\right) 
                        +\Delta(\nT)$$
vertices in $X$. If $\sP[\nT]=\cT[\nT]$ or $\sP[\nT]=\alpha$, then $C_\nT \subseteq X$, and hence any vertex $v \in C_\nT$ is a neighbor of $$\sum_{\nT' \in \ntypes} |X \cap C_{nT'}| -1 + |\nT^* \cap \nT|= $$$$\sum_{\substack{\nT \in \ntypes :\\ (\sP[\nT] < \alpha / 2 \\ \text{ or } \cT[\nT] \leq \alpha)}} \sP[\nT]
  +\sum_{\substack{\nT \in \ntypes :\\ (\cT[\nT]=\alpha+1 \text{ and } \\  \sP[\nT] > \alpha / 2)}} \left( |C_\nT| - \alpha + \sP[\nT]\right) 
                        +\Delta(\nT)$$ vertices in $X$, as we consider here open neighborhoods.  Since the fair cost maximizes the number of vertices of $X$ in the neighborhood, we consider the maximum $\Delta(\nT)$. 
\end{proof}
}

\begin{lemma}\label{lem:coherent_shapeILP}
  Let $G$ be a graph and $\phi(\aleph)$ be a formula with one free variable.
  Let $f$ be the function from \cref{lem:shape_equivalence_one} and let $\alpha, \gamma>f(|\MSs|,|\phi|)$ be integers.
  Then there is an \FPT algorithm parameterized by
  $|\phi|, |\MSs|,\alpha,\gamma$ which, for every coherent
  shape $\shape = (\nT^*, \mtx) \in \shapes^{(\alpha, \gamma, |\MSs|)}$
  and any $k\in\mathbb{N}$, decides if there exists an
  $\alpha$-compliant set $X \subseteq V(G)$ such that
  $G \models \phi(X)$, $\shapefn^{(\alpha,\gamma)}(G,X)=\shape$,
  and $\fc{X} \leq k$.
\end{lemma}

\begin{proof}
We start with a description of the algorithm and later prove its correctness and time complexity.
First, we need to check whether $\shape$ satisfies the formula. If not, we can already conclude that such a set $X$ does not exist, since any $X$ that agrees with $\shape$ satisfies $G \models \phi(X)$ if and only if $\models \phi(\shape)$ by \Cref{lem:shape_equivalence_one}.
Assume now that $\shape$ satisfies the formula $\phi(\aleph)$.

Let $\cT \in\ctypes^{(\alpha+1,|\MSs|)}$ and denote $\Cc_G$ by $\Cc$.
By $\Cc_{\cT}$ we denote the set of cliques from $\Cc$ whose
$(\alpha+1)$-clique type is $\cT$.
The shape $\shape$ provides us with
solution patterns
$\spatterns^{(\alpha, |\MSs|)}$ used to index the matrix $\mtx$.
For conciseness, in this proof, we omit the superscript $(\alpha, |\MSs|)$ as it is clear from the context.
For each $\mathbf{S} \subseteq \spatterns$, let
\[
  \Cc_{\cT}^{\mathbf{S}} = \left\{C \in \Cc_{\cT} : \{\sP :  \cliFC(C, \sP, \nT^*) \leq k \}  = \mathbf{S}  \right\}.
\]

Informally speaking, $\Cc_{\cT}^{\mathbf{S}}$ is a set of cliques of type $\cT$ that have fair cost at most $k$ on precisely the solution patterns from $\mathbf{S}$.
Observe that the set $\Cc_{\cT}^{\mathbf{S}}$
partitions the cliques in $\Cc_\cT$ according to $\mathbf{S}$.
Introducing sets $\mathbf{S}$ allows us to stop worrying about the fair cost of vertices within the cliques. It remains to enforce that the fair cost of vertices in the modulator is at most $k$.
\lv{

}%
Let $X_{\cT}^{\mathbf{S}, \sP}$ be a variable representing the number of cliques from $\Cc_{\cT}^{\mathbf{S}}$ that use $\sP$ as a solution pattern. We formulate \cref{il_faircost} for all subsets $\mathbf{S}$ of $\spatterns$.
We also need to ensure that the variables $X_{\cT}^{\mathbf{S}, \sP}$ are compatible with $\shape$, which is ensured by \cref{ilp_compile1} and \cref{ilp_compile2}.

\noindent
\begin{align}
    \sum_{\sP \in \mathbf{S}}  X^{\mathbf{S}, \sP}_{\cT} = |\Cc_{\cT}^{\mathbf{S}}| && \forall_{\cT \in \ctypes} \forall_{\mathbf{S} \subseteq \spatterns } \label{il_faircost} \\
    \sum\limits_{\mathbf{S}:\sP\in\mathbf{S} }  X^{\mathbf{S}, \sP}_{\cT} = \mtx[\sP, \cT] && \forall_{\cT \in \ctypes} \forall_{\sP \in \spatterns } : \mtx[\sP, \cT] < \gamma \label{ilp_compile1} \\ 
    \sum\limits_{\mathbf{S}:\sP\in\mathbf{S} }  X^{\mathbf{S}, \sP}_{\cT} \geq \gamma  && \forall_{\cT \in \ctypes} \forall_{\sP \in \spatterns } : \mtx[\sP, \cT] = \gamma \label{ilp_compile2}
\end{align}

We have ensured that the solution will comply with $\shape$ and that the fair cost of every vertex inside any clique is at most $k$. 
It remains to bound the fair cost for every $v \in \MSs$. To do so, we need to distinguish two cases, as we only consider coherent shapes. 
For all cliques $C$ of a specific type $\cT$, and for every $\nT$,
we know whether we take almost all or almost none vertices
of $C_\nT$ as described in the matrix $\mtx$ in $\shape$. This allows us to count the fair cost contributions of the vertices in the modulator.
\lv{

}%
Note that the variables $X_{\cT}^{\mathbf{S}, \sP}$ assign the cliques of type $\cT$ to solution patterns from $\spatterns$. For a vertex $w$ of a clique $C$ of type $\cT$ let $\nT(w) \in \ntypes$ be such that $w \in C_{\nT(w)}$. 
We call a clique vertex $w$ \emph{bounded} if $(\cT,\nT(w))$ is bounded. We call a vertex $w$ \emph{thin} (resp.\ \emph{thick}) if its clique of clique type $\cT$ is assigned to a solution pattern $\sP$ such that $(\nT(w),\cT,\sP)$ is thin (resp.\ thick). 
Let us fix a vertex $v$ from the modulator $\MSs$. 
We first calculate the number of neighbors of $v$ whose type 
is bounded.
\[
  \mathsf{bounded}(v) = \sum_{\cT \in \ctypes} \sum_{ \sP \in \spatterns } \sum_{\substack{\nT \in \ntypes, \\ \nT[v]=1 , \\ \cT[\nT] \leq \alpha }}  \sum_{\substack{ \mathbf{S} \subseteq \spatterns, \\ \sP \in \mathbf{S}}}  \big( X^{\mathbf{S}, \sP}_{\cT}\cdot\sP[\nT] \big).
\]

Now, we calculate the number of thin neighbors of $v$:
\[
  \mathsf{thin}(v) = \sum_{\cT \in \ctypes} \sum_{ \sP \in \spatterns } \sum_{\substack{\nT \in \ntypes, \\ \nT[v]=1, \\ \cT[\nT]=\alpha+1, \\ \sP[\nT] < \alpha / 2} }  \sum_{\substack{ \mathbf{S} \subseteq \spatterns, \\ \sP \in \mathbf{S}}}  \big( X^{\mathbf{S}, \sP}_{\cT}\cdot\sP[\nT] \big).
\]

To calculate the number of thick neighbors of $v$, we first sum the sizes of all $C_\nT$ where the shape is thick, and next we subtract the number of vertices excluded from $X$ in thick parts:
\[
  \mathsf{thick\_all}(v) = \sum_{\cT \in \ctypes} \sum_{C \in \Cc_\cT } \sum_{\substack{\nT \in \ntypes, \\ \nT[v]=1, \\ (\nT, \cT) \text{ is thick}} }  |C_\nT|.
\]

Now, we can calculate the contribution of thick solution patterns as follows:
\[
  \mathsf{thick}(v) = \mathsf{thick\_all}(v) - \sum_{\cT \in \ctypes} \sum_{ \sP \in \spatterns } \sum_{\substack{\nT \in \ntypes, \\ \nT[v] = 1, \\ \cT[\nT]=\alpha+1, \\ \sP[\nT] > \alpha / 2} }  \sum_{\substack{ \mathbf{S} \subseteq \spatterns, \\ \sP \in \mathbf{S}}} \big( X^{\mathbf{S}, \sP}_{\cT}\cdot(\alpha - \sP[\nT]) \big).
\]

\noindent
From that, we can calculate $\fc{v}$ for any $v \in \MSs$ as:
\[
  \text{FC}(v) =\mathsf{bounded}(v) + \mathsf{thick}(v)+\mathsf{thin}(v) + |N(v)\cap \nT^*|.
\]
Hence, we arrive at the final set of constraints:
\begin{align}
  \text{FC}(v) \leq k && \forall_{v \in \MSs}. \label{il_modulator}
\end{align}
To summarize, the whole ILP consists of Constraints 
(\ref{il_faircost}),
(\ref{ilp_compile1}),
(\ref{ilp_compile2}), 
and (\ref{il_modulator}).

We now analyze the time complexity. 
The only variables in the ILP are $X_{\cT}^{\mathbf{S}, \sP}$. 
Observe that $|\ctypes|$, $|\spatterns|$, and the number of subsets $\mathbf{S}\subseteq \spatterns$ are bounded by functions of $\alpha$, $\gamma$, $|D_G|$, and $|\phi|$.
Hence, we have a bounded number of variables, and we can solve the ILP in the stated running time by applying the algorithm of \cref{thm:ILPinBoundedDimension}.
Now, we prove correctness of the proposed algorithm. Let $I_{G, \phi}^\shape$ be the ILP instance for a graph $G$, formula $\phi$, and shape $\shape$. We need to prove that $I_{G, \phi}^\shape$ is a YES-instance if and only if there exists $X \subseteq V(G)$ such that $\fc{X} \leq k$, $G \models \phi(X)$, $X$ is $\alpha$-compliant, and $X$ agrees with $\shape$.

We start with the ``only if'' direction. 
We construct a set $X$ as follows. For each $\cT$ and each $\mathbf{S}$, we assign exactly $X^{\mathbf{S}, \sP}_{\cT}$ cliques in $\Cc_\cT^\mathbf{S}$ the solution pattern $\sP$. This defines a set $X$. Now, $\fc{X} \leq k$ follows from the fact that $I_{G, \phi}^\shape$ is a YES-instance and from the construction of $\mathbf{S}$, which ensures that the fair cost of vertices both in $D_G$ and in $\Cc$ is at most $k$. 
Observe that agreement of $X$ with $\shape$ is ensured by \cref{il_faircost}, \cref{ilp_compile1}, and \cref{ilp_compile2}.
We know that $G \models \phi(X)$ because feasibility of the ILP implies $\models \phi(\shape)$, and from \cref{lem:shape_equivalence_one} we deduce that $G \models \phi(X)$. The fact that $X$ is $\alpha$-compliant follows directly from the construction.

For the ``if'' direction, let $X \subseteq V(G)$ be an $\alpha$-compliant set that satisfies all assumptions and let $\shape$ be a shape agreeing with $X$. Because the sets $\Cc_\cT^\mathbf{S}$ partition cliques, we can set each variable $X^{\mathbf{S}, \sP}_{\cT}$ to be the number of cliques in $\Cc_\cT^\mathbf{S}$ on which $X$ matches $\sP$. Verifying that all constraints are satisfied is analogous to the previous direction. %
\end{proof}

We prove the main theorem of this section (\cref{thm:main-fpt}) by an iterative application of \cref{lem:coherent_shapeILP}.
We can compare the following approach to the proof of the model-checking corollary (\cref{cor:mc}\sv{ in Appendix \ref{sec:mc-app}}).

\begin{proof}[Proof of \cref{thm:main-fpt}]
    Let $f$ be the function from \cref{lem:shape_equivalence_one}.
    We iterate over each $\shape\in\shapes^{(\alpha, \gamma, |\MSs|)}$.
    We model-check the formula on $\shape$, which leads to model-checking of the associated graph and the set (consult \cref{def:associatedGraph}), whose sizes are bounded by the parameters.
    If $\models \phi(\shape)$ then we use \cref{lem:coherent_shapeILP} to decide, in \FPT time, whether there exists an $\alpha$-compliant set $X$ agreeing with $\shape$ and of fair cost at most $k$.
    This proves that we are able to solve the \textsc{\MSOo-FairVE} problem in \FPT time.
\end{proof}

\toappendix{
  \section{Problems Covered by {\protect\Cref{thm:main-fpt}} (Proof of {\protect\Cref{cor:problmesFPT}})}\label{sec:problems}

In this short section, we provide a list of particular problems that are covered by the rather technical statement of \Cref{thm:main-fpt}.
\CorSlvdProblems*

Alongside the definitions, we also give justifications of why their shapes are coherent.
Therefore, we prove \cref{cor:problmesFPT}.

\prob{\sc Fair Vertex Cover}
{An undirected graph $G$ and an integer $k$.}
{Is there a set $X \subseteq V(G)$ of fair cost at most $k$ such that every edge has at least one endpoint in $X$?}

\prob{\sc Fair Feedback Vertex Set}
{An undirected graph $G$ and an integer $k$.}
{Is there a set $X \subseteq V(G)$ of fair cost at most $k$ such that $G-X$ is a forest?}

\prob{\sc Fair Odd Cycle Transversal}
{An undirected graph $G$ and integer $k$.}
{Is there a set $X \subseteq V(G)$ of fair cost at most $k$ such that $G - X$ is a bipartite graph?}

The crucial observation is that solving all the aforementioned problems on cliques is relatively straightforward.
Any valid solution forces the set $X$ to contain all but at most a constant number of vertices in each clique (or in each $C_\nT$).
We formulate this idea as the following observation.
\begin{observation}[Solution structure on a clique]\label{obs:solutionStructure}
    Let $G$ be any graph. For any problem $\pi$ defined above, there is a constant $c_\pi$ such that the following holds.
    Whenever we have $X \subseteq V(G)$ that satisfies $\pi$ (except for the fair cost condition), then for any $C \in \Cc$, it holds that $|X \cap C| \geq |C|-c_\pi$. In particular, for every $\nT$, it holds that $|X \cap C_\nT| \geq |C_\nT|-c_\pi$.
\end{observation}
\Cref{obs:solutionStructure} immediately leads to coherent shapes, as all non-zero entries of the shape matrix are thick or bounded.
It is a routine and standard check that the above problems can be expressed by an \MSOo formula of constant size.
Therefore, we conclude by \Cref{thm:main-fpt}.
Observe that $c_\pi\le2$ for all the problems mentioned above.
Hence, we prove \cref{cor:problmesFPT} for them.

We now move to problems where the coherence originates in the matrix entries being thin.
We can formulate an analogous observation where we replace the last condition with $|X \cap C_\nT| \leq c_\pi$.

\prob{\sc Fair [$\sigma,\rho$]-Domination}
{An undirected graph $G=(V,E)$, finite or cofinite $\sigma, \rho\subseteq \mathbb{N}$, and an integer $k$.}
{Is there a set of vertices $X\subseteq V(G)$ of fair cost at most $k$ such that for all $v\in X$, $|N(v)\cap X|\in\sigma$ and for all $v\in V(G)\setminus X$, $|N(v)\cap X|\in\rho$?}
Observe that \textsc{Fair Dominating Set} is a special case of $[\sigma,\rho]$-Domination where $\sigma=\mathbb{N}$ and $\rho=\mathbb{N}\setminus\{0\}$.
Again, it is a routine check to describe $[\sigma,\rho]$-Domination by an \MSOo formula, depending on $\sigma,\rho$.

We assume that $\sigma$ is finite. 
We define $\max_\sigma\df \max \sigma+1$.
Observe that each clique contains at most $\max_\sigma$ vertices of the solution.
Therefore, we conclude that any such shape is coherent, with all non-zero entries of the associated matrix thin or bounded.

We assume that both $\sigma$ and $\rho$ are cofinite. 
We define $\max_{\sigma,\rho}\df \max_{\substack{ x\in\mathbb{N}, x\not\in \rho, \\ x \not\in \sigma}} (x+1)$.
We argue that it is never profitable to put into $X$ more than $\max_{\sigma,\rho}$ vertices of any $C_\nT$.
If there are more than $\max_{\sigma,\rho}$ vertices in some $C_\nT\cap X$, we make $|C_\nT\cap X|=\max_{\sigma,\rho}$.
Indeed, this operation violates neither the $\rho$ nor the $\sigma$ conditions. For any vertex $v$ in a clique $C$, either the value $|N(v) \cap X|$ remains unchanged or it still belongs to both $\rho$ and $\sigma$. The same holds for any vertex $v$ from the modulator. %
Moreover, this reduction does not increase the fair cost, as taking subsets can only decrease it.
Hence, we conclude that such shapes are coherent, with all non-zero entries of the associated matrix thin or bounded.

By this analysis we conclude the proof of \cref{cor:problmesFPT}.

}

\toappendix{
\section{Hardness}\label{sec:hardness}

In this section, we present the reduction proving \W{1}-hardness of \FairVD, stated in the following theorem.
\WhThm*

We use the \textsc{Unary $\ell$-Bin Packing} problem as the starting point of our hardness reduction.
The \textsc{Unary $\ell$-Bin Packing} problem is \W{1}-hard parameterized by $\ell$, the number of bins~\cite{JansenKMS13}.
There, the item sizes are encoded in unary and the task is to assign $n$ items to $\ell$ bins such that the sum of sizes of items assigned to any bin does not exceed its capacity $B$.
Formally, \textsc{Unary $\ell$-Bin Packing} is defined as follows.

\prob{\textsc{Unary $\ell$-Bin Packing}\hspace{20em} {\em Parameter:} $\ell$}
{Positive integers $\ell, B$ and a multiset $\mathcal{S}$ of items of non-zero sizes $s_1, \ldots, s_n$ encoded in unary.}
{Is there a packing of all items into at most $\ell$ bins of size $B$?

More formally, is there a function $f: [n] \rightarrow [\ell]$ such that
\[
  \underset{k \in [\ell]}{\forall}  \underset{i \in f^{-1}(k)}{\sum} s_i \leq B \, ?
\]}

Now we introduce the first problem, which we reduce from \dbinpacking.

\prob{\textsc{Unary $d$-Tuple}\hspace{18em} {\em Parameter:} $d$}
{Positive integers $b, d$ and a multiset $\mathcal{A}$ of $n$ integer tuples $(a^{(i)}_1, \ldots, a^{(i)}_d)$, $i \in \{ 1, \ldots ,n \}$, where all entries are encoded in unary.}
{Is there a function $f : [n] \rightarrow [d]$ such that
\[
  \underset{k \in [d]}{\forall}  \sum_{i \in f^{-1}(k)} a^{(i)}_k \leq b \, ?
\]}

\noindent
The main result of this section is based on the two lemmas that follow next. In the first lemma (\cref{lem:tuple-W1}) we show a reduction from the \dbinpacking problem to the \dtuple problem, which implies \W{1}-hardness of the latter.

\begin{lemma}\label{lem:tuple-W1}
  \dtuple problem is \W{1}-hard parameterized by $d$.
\end{lemma}

\begin{proof}
  To show \W{1}-hardness of \dtuple we reduce from \dbinpacking.
  Let us take any instance $(\mathcal{S}, \ell, B)$ of \dbinpacking.
  We construct an instance $(\mathcal{A}, d, b)$ of \dtuple by setting $b = B$, $d = \ell$, and
  \[
    \forall i \in [n]\  \forall k \in [d] : \ a^{(i)}_k = s_i \, .
  \]

  Assume that $(\mathcal{A}, d, b)$ is a YES-instance, and let $f$ be a witnessing function.
  Since $d = \ell$, $b = B$, and $a^{(i)}_k = s_i$ for all $i$ and for all $k\in[d]$, we can rewrite the condition
  \[
    \forall k \in [d] : \sum_{i \in f^{-1}(k)} a^{(i)}_k \leq b
  \]
  as
  \[
    \forall k \in [\ell] : \sum_{i \in f^{-1}(k)} s_i \leq B \, .
  \]
  Hence, it is easy to observe that every witness for a YES-instance of the \dtuple problem is also a witness for a YES-instance of the \dbinpacking problem, and vice versa.
\end{proof}

Now, we show the main reduction, reducing from \dtuple and proving \cref{thm:hardness}; see \Cref{fig:reduction} for a visual representation of the reduction.

\begin{figure}
\centering
\begin{subfigure}{0.7\textwidth}
    \includegraphics[width=\textwidth]{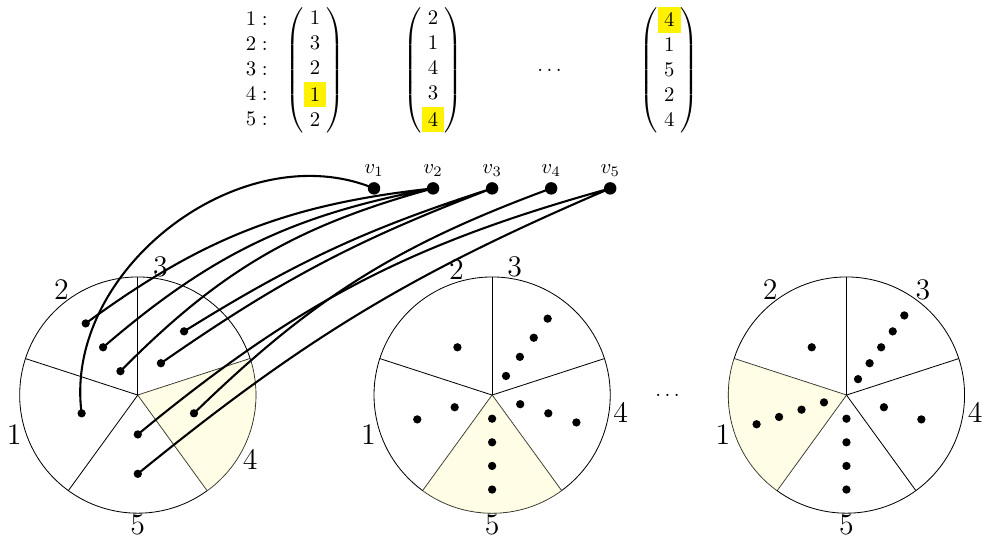}
\end{subfigure}

\caption{At the top, there is an instance of \dtuple (with the solution highlighted and $d=5$). Below, there is an instance of \FairVD, with $\cvd(G)=5$, but without private cliques depicted. Hence, there are as many cliques as there are tuples in the \dtuple instance. Each clique corresponds to one tuple, and the number of vertices connected to the $i$-th vertex of the modulator equals the $i$-th value in the tuple. For better readability, we only present the edges between the first clique and the modulator in the picture.}
\label{fig:reduction}
\end{figure}

\begin{proof}[Proof of \cref{thm:hardness}]
  We reduce from \dtuple, which was shown \W{1}-hard in \cref{lem:tuple-W1}.

  We describe the construction from \dtuple to \FairVD.
  Let $(\mathcal{A}, d, b)$ be an instance of \dtuple.
  First, we remove all tuples that have at least one zero entry.
  This is safe, since such a tuple can always be assigned to a position where it contributes $0$ to the corresponding sum constraint.

  We construct a graph $G$ as follows.
  First, we create vertices $v_1, v_2, \dots, v_d$.
  For the $i$-th tuple $(a^{(i)}_1, \ldots, a^{(i)}_d) \in \mathcal{A}$, we construct a clique $C^i$ of size $\sum_{k=1}^d a_k^{(i)}$, and we connect exactly $a_k^{(i)}$ vertices of this clique to vertex $v_k$ for each $k \in [d]$.
  This is done in such a way that every vertex in the clique $C^i$ is adjacent to exactly one vertex of the modulator.
  We denote the set of those $a_k^{(i)}$ vertices from clique $C^i$ by $C^i_{k}$.

  Additionally, we create a collection of \emph{private cliques} that plays a role of preserving $\mathsf{modul}$ (defined in the next paragraph) after deletion.
  For each $k \in [d]$, we additionally create three cliques $P^k_1, P^k_2, P^k_3$, each of size $b+1$, and we connect every vertex of each $P^k_t$ to $v_k$ (and to no other $v_{k'}$).
  In case there is $k\in[d]$ for which $\sum_{i\in[n]}a_k^{(i)}\le b$, we answer YES as this is a trivial instance.
  Hence, we know that $b$ is comparable with item sizes encoded unary on the input.
  Again, every vertex outside $\{v_1,\dots,v_d\}$ is adjacent to exactly one vertex of the modulator.
  Observe that $\{v_1, v_2, \dots, v_d\}$ is still a modulator of size $d$, since removing it leaves us with a set of vertex-disjoint cliques (the cliques $C^i$ and the cliques $P^k_t$).

  Now we propose the formula for the instance of \FairVD.
  First, we define an \FO formula $\mathsf{modul}(v)$, which checks whether the vertex $v$ belongs to the modulator.
  Observe that, due to our construction with private cliques, for any deletion set $X$ with $\fc(X) \le b$, each $v_k$ vertex in modulator keeps at least one neighbor in each of $P^k_1, P^k_2, P^k_3$, and thus $\mathsf{modul}(v_k)$ remains true in $G - X$.
  In contrast, no vertex outside the modulator satisfies this property.
  The above property can be encoded by the following \FO formula:
  \[
    \mathsf{modul}(v) \df \exists_{a, b, c \in N(v)} \ a \neq b \wedge a \neq c \wedge b \neq c \wedge \neg e(a, b) \wedge \neg e(b, c) \wedge \neg e(c, a).
  \]

  We then present the final formula, which is satisfied if and only if, in every clique $C^i$, we remove at least one set $C^i_k$ for some $k \in [d]$, while leaving the modulator unchanged.
  Observe that the above is true if and only if there are \textbf{no} $d$ vertices outside the modulator that lie in one clique and each of them is a neighbor of a different vertex of the modulator.
  We encode this by the following \FO sentence:
  \[
  \begin{gathered}
      \phi \df \exists_{m_1, \dots, m_d}
      \bigwedge_{ \substack{ i,j \in[d] \\ i \neq j}} m_i \neq m_j
      \ \wedge\ \bigwedge_{i \in [d]} \mathsf{modul}(m_i) \ \wedge \\
      \neg \Bigg( \exists_{v_1, \dots, v_d}
      \bigwedge_{ \substack{ i,j \in[d] \\ i \neq j}} v_i \neq v_j
      \ \wedge\ \bigwedge_{i \in [d]} \neg \mathsf{modul}(v_i)
      \ \wedge\ \bigwedge_{\substack{ i,j \in[d] \\ i \neq j}} e(v_i, v_j)
      \ \wedge\ \bigwedge_{\substack{ i,j \in[d] \\ i \neq j}} \Big(\exists u : \mathsf{modul}(u) \wedge e(u, v_i) \wedge \neg e(u, v_j)\Big) \Bigg).
  \end{gathered}
  \]
  Finally, we set the value of the fair-cost threshold in \FairVD to be $k \df b$.

  We now analyze the correctness of the reduction.
  Since entries of \dtuple are encoded in unary, the number of vertices in $\bigcup_i C^i$ is linear in the sum of all tuple entries, and the private cliques contribute only $3d(b+1)$ additional vertices (which is comparable to the unary encoded items due to preprocessing of the trivial instances); hence $|G|$ is polynomial in the input size.
  The size of $\phi$ depends only on $|\MSs|$.

Finally, we show the completeness and soundness of the reduction.
Assume that \dtuple is a YES-instance, and let $f$ be a witness of that. 
Then the set $X$ such that $\forall_{i \in [n]}  X \cap C^i = C^i_{f(i)}$ is a witness for a YES-instance of \FairVD.

Now, assume that \FairVD is a YES-instance and let $X \subseteq V(G)$ be a witness of that.
Note that $X \cap \cvd(G) = \emptyset$.
Moreover, for every $i \in [n]$ there exists some $k \in [d]$ such that $C^i_k \subseteq X$, since $X$ satisfies $\phi$.
Let $X' \subseteq X$ be such that for every $i \in [n]$ we have $X' \cap C^i = C^i_k$ for some $k \in [d]$ with $C^i_k \subseteq X$.
Observe that $X'$ still satisfies $\phi$ and $\fc{X'} \leq \fc{X} \leq k$.
Let $f$ be such that for every $i \in [n]$, $f(i)$ is equal to $k \in [d]$ if $X' \cap C^i = C^i_k$.
Then $f$ is a witness that the instance of \dtuple is a YES-instance.
\end{proof}
}

\lv{
\section{Conclusions}\label{sec:conc}
In this paper, we resolved an open question posed in \cite{KMT19}.
We showed that, despite the general problem being \W{1}-hard, we are able to solve many graph vertex problems in \FPT time.
It would be interesting to see whether we can extend (some of) the \FPT results to denser graph parameters.
In particular, \cite{KMT19} shows that the Fair Vertex Cover problem is solvable in \FPT time parameterized by the modular-width.
We would like to emphasize their question of whether $\MSOo$-FairVE is \FPT parameterized by the modular-width and the size of the formula.

While \cref{thm:main-fpt} covers a large number of natural graph vertex problems, it would be interesting to know whether all cases of Fair $[\sigma,\rho]$-Domination problems are \FPT parameterized by the cluster vertex deletion; namely, the unresolved case where $\sigma$ is cofinite and $\rho$ is finite.

More generally, our intuition is that the tractable cases of fair vertex problems parameterized by the cluster vertex deletion are close to \cref{thm:main-fpt}.
We leave it as an open and challenging question where the dichotomy is drawn.

Finally, we propose the previously overlooked question of whether \MSOo-FairVE problems are in \XP parameterized by the clique-width and the size of the formula.
For the \MSOt-FairVE problem, \cite{Kolman09onfair} showed an \XP algorithm parameterized by the treewidth and the size of the formula, as observed by \cite{MT20}.
This result was later implied by a more general result of \cite{KKMT}.
However, to the best of our knowledge, an \XP algorithm parameterized by clique-width is missing.
}

\bibliography{main}

\sv{
  \newpage
\appendix
\section{Main Tools and Model-checking Machinery---Proof of \cref{lem:shape_equivalence_one}} \label{sec:mc-app}
\appendixText
}

\end{document}